%% file: MisoLowCSITjournalV12.tex
\documentclass[10pt,journal,a4paper, final, romanappendices]{IEEEtran}

\makeatletter
\newenvironment{varsubequations}[1]
 {%
  \addtocounter{equation}{-1}%
  \begin{subequations}
  \renewcommand{\theparentequation}{#1}%
  \def\@currentlabel{#1}%
 }
 {%
  \end{subequations}\ignorespacesafterend
 }
\makeatother

\usepackage[cmex10]{amsmath}

\usepackage{amssymb}
\usepackage{amsthm}
	\usepackage{bm}
	\usepackage{nicefrac}
	\usepackage{dsfont} 
	\usepackage[subnum]{cases}  

\usepackage[utf8]{inputenc}
\usepackage[T1]{fontenc}
\usepackage{url}
\usepackage{ifthen}
\usepackage{cite}

\usepackage{graphicx}
\newtheorem{theorem}{Theorem}
\newtheorem{definition}{Definition}
\newtheorem{lemma}{Lemma}
\newtheorem{corollary}{Corollary}
\newtheorem{remark}{Remark}

\graphicspath{{./figures/}}
	\usepackage{cite}
	\usepackage[ruled,vlined,linesnumbered]{algorithm2e}
	
	\newcommand\blfootnote[1]{%
  \begingroup
  \renewcommand\thefootnote{}\footnote{#1}%
  \addtocounter{footnote}{-1}%
  \endgroup
}

\usepackage{xcolor}
\usepackage{soul}
\usepackage[multiple]{footmisc}

\usepackage{mathtools}

\newcommand{\Expect}{{\rm I\kern-.3em E}}

\newcommand{\bfd}{{\mathbf{d}}}
\newtheorem{ex}{{\em Example}}

\input{Definitions}

\begin{document}
\pagenumbering{gobble}

\title{Resolving the Feedback Bottleneck of \\ Multi-Antenna Coded Caching}

\author{Eleftherios Lampiris, Antonio Bazco-Nogueras, Petros Elia\\ lampiris@tu-berlin.de, \{bazco, elia\}@eurecom.fr}

\maketitle
\nocite{maddah2014fundamental}
\begin{abstract}\blfootnote{E. Lampiris is with the Communications and Information Theory Group (CommIT) of the Technical University of Berlin, 10553 Berlin, Germany and A. Bazco-Nogueras and P. Elia are with the Communication Systems Department of EURECOM, 06410 Sophia Antipolis, France. This work was partially supported by the ANR project ECOLOGICAL-BITS-AND-FLOPS and by the European Research Council under the EU Horizon 2020 research and innovation program / ERC grant agreement no. 725929. (ERC project DUALITY) and by the European Research Council under the ERC grant agreement N. 789190 (project CARENET). This work was conducted while the first author was at EURECOM. Part of this work was presented in the 2018 IEEE International Symposium on Information Theory (ISIT) \cite{lampiris2018lowCSIT}.}
Multi-antenna cache-aided wireless networks have been known to suffer from a severe feedback bottleneck, where achieving the maximal Degrees-of-Freedom (DoF) performance required feedback from all served users. These costs matched the caching gains and thus scaled with the number of users.
In the context of the $L$-antenna MISO broadcast channel with $K$ receivers having normalized cache size $\gamma$, we pair a fundamentally novel algorithm together with a new information-theoretic converse, and identify the optimal tradeoff between feedback costs and DoF performance, by showing that having CSIT from only $C<L$ served users implies an optimal one-shot linear DoF of $C+K\gamma$. As a side consequence of this, we also now understand that the well known DoF performance $L+K\gamma$ is in fact exactly optimal.
In practice, the above means that we are now able to disentangle caching gains from feedback costs, thus achieving unbounded caching gains at the mere feedback cost of the multiplexing gain. This further solidifies the role of caching in boosting multi-antenna systems; caching now can provide unbounded DoF gains over multi-antenna downlink systems, at no additional feedback costs. The above results are extended to also include the corresponding multiple transmitter scenario with caches at both ends.
\end{abstract}

\section{Introduction}

The seminal work of Maddah-Ali and Niesen \cite{maddah2014fundamental} revealed how caching modest amounts of content at the receivers has the potential to yield unprecedented reductions in the delivery delay of content-related traffic.

Specifically, the work in \cite{maddah2014fundamental} considered a shared-link broadcast channel, where a transmitter is tasked with serving content from a library of $N$ files to $K$ receiving users. Each user is endowed with a cache that can store a fraction $\gamma\in[0,1]$ of the library, thus yielding a cumulative cache size of
$t\triangleq K\gamma$, which essentially means that each part of the library can appear $t$ different times across the different caches. The approach of \cite{maddah2014fundamental} was to design the cache placement algorithm in such a manner that desired content that resides in different  caches could be combined together to form a single transmitted multicast signal that carries information for multiple users. In turn, these same users would then access their individual caches in order to remove all the unwanted interference from the multicast signal, and thus decode their desired message. 
In this shared-link (noiseless, wired) setting, with unitary link capacity, this strategy allows for a worst-case (normalized) delivery time of
\begin{equation}
	\mathcal{T}_{1}(t) = \frac{K-t}{1+t},
\end{equation}
which implies an ability to serve $1+t$ users at a time.
This performance is shown in \cite{yuTradeoff2TransIT2019} to be within a multiplicative gap of $2.01$ of the optimal gain, 
while under the assumption of uncoded placement the above performance is exactly optimal \cite{wanOptimalityTransIT2020,YuMA18}.

The direct extension of this result to the equivalent high Signal-to-Noise Ratio (high-SNR) single-antenna \emph{wireless} Broadcast Channel (BC) --- where similarly the long-term capacity of each point-to-point link is normalized to $1$ file per unit of time --- implies a Degrees-of-Freedom\footnote{We properly define the DoF in Section~\ref{se:converse_0}, Definition~\ref{def:dof}.} (DoF) performance of
\begin{equation}
	\mathcal{D}_{1}(t) \triangleq \frac{K-t}{\mathcal{T}_{1}(t)} = 1+t,
\end{equation}
which can be achieved without any Channel State Information at the Transmitter (CSIT).

This came in direct contrast with multi-antenna systems which are known to also provide DoF gains but only with very high feedback costs that scale with these DoF gains. As it is known (cf.~\cite{GrassmanTse2002}, \cite{lozanoFundCoop}), such feedback costs are the reason for which most multi-antenna solutions fail to scale (cf.~\cite{caire2003achievable,jafar2005isotropic,huang2012degrees,lapidoth2006capacity,vaze2012degree,maddah2012completely,Bazco2018KmisoBC,jafar2012blind,yang2013degrees,gou2012optimal,chen2012degrees,chen2013toward,tandon2013synergistic,chen2015two}). The huge impact of feedback on the network's performance has triggered a major research interest in understanding how imperfect, partial, or limited feedback can help improve system performance~\cite{Love2008_limitedFeedbackOverview_JSAC}.
Among the vast literature that resulted from this interest, different works have focused, for example, on analyzing the impact of feedback in interference-limited multi-antenna cellular networks \cite{Lee2011_limitedFeedbackIBC_TWC, Park2016_optFeedbackCellular_TWC,Min2018_optFeedbackCellular_TWC}, the  feasibility of Interference Alignment\cite{ElAyach2012_IAfeedback_TWC}, the limited-feedback resource allocation in heterogeneous wireless networks \cite{Mokari2016_limFeedbackHetNet_TransVehicTech}, the capacity for Gaussian multiple access channels with feedback~\cite{Sula2020capacityGMAC}, and the effect of either rate-limited feedback in the interference channel \cite{Vahid2012_limitedFeedback_TIT} or SNR-dependent feedback in the Broadcast Channel\cite{davoodi2016gdof}.
Recently, the analysis of the significance of feedback has been extended also to secure communications~\cite{Bassi2019wiretap} as well as to the capacity of burst noise-erasure channels~\cite{Song2019burstErasureFeedback}.

\subsection{Multi-antenna cache-aided channels} At the same time, there is substantial interest in combining the gains from caching with the traditional multiplexing gains of feedback-aided multi-antenna systems. Combining the two ingredients is only natural, given the promise of coded caching and the fact that multi-antenna technologies are currently the backbone of wireless systems. One can argue that coded caching stands a much better chance in becoming a pertinent ingredient of wireless systems, if it properly accounts for the fact that the most powerful and omni-present resource in current networks is multi-antenna arrays.

This direction seeks to merge two seemingly opposing approaches, where traditional feedback-based multi-antenna systems work by creating parallel channels that separate users' signals, while coded caching fuses users' signals and counts on each user receiving maximum interference.
In this context, the work in \cite{shariatpanahiMultiserverTransIT2016} analyzed the wired multi-server ($L$ servers) setting, which can easily be seen to correspond to the high-SNR cache-aided MISO BC setting with $L$ transmit antennas. An interesting outcome of this work is the revelation that multiplexing and caching gains can be combined additively, yielding an achievable DoF equal to
\begin{equation}\label{eqMSDoF}
	\mathcal{D}_{L}(t)=L+t.
\end{equation}
In the same spirit, the work in \cite{7857805} studied the $K_{T}$-transmitter, fully-connected network where the transmitters are equipped with caches that can each store a fraction $\gamma_{T}\in[1/K_T,1]$ of the library, amounting to a cumulative (transmitter-side) cache size of $t_{T} = K_{T}\gamma_T$. Under a cumulative receiver-side cache size of $t$, the achievable DoF this time took the form
\begin{equation}\label{eqInterferenceDoF}
	\mathcal{D}_{t_{T}}(t)=t_{T}+t.
\end{equation}

As shown in \cite{7857805}, under the assumption of uncoded placement, the performance in~\eqref{eqMSDoF}-\eqref{eqInterferenceDoF} is at a factor of at most 2 from the optimal one-shot linear DoF.
Since then, many works such as \cite{shariatpanahi2017multi,tolli2017multi,roigCachesBothEndsICC2017,8007039,lampiris2018adding,shariatpanahi2018multi,lampirisBridgingAsilomar2019,lampirisBridgingSPAWC2019} have developed different coded caching schemes for the multi-transmitter and the multi-antenna settings.

\subsection{Scaling feedback costs in multi-antenna coded caching}

While the single antenna case in~\cite{maddah2014fundamental} provides the near optimal caching gain $t$ without requiring any CSIT, a significant feedback problem arises in the presence of multiple antennas.
Specifically, all known multi-antenna coded caching methods \cite{shariatpanahiMultiserverTransIT2016,shariatpanahi2017multi,tolli2017multi} that achieve the full DoF $L+t$ require each of the $L+t$ benefiting receivers to communicate feedback to the transmitter. To make matters worse, the problem extends to the dissemination of CSI at the receivers (CSIR), where now the transmitter is further forced to incorporate in this CSIR additional information on the CSIT-based precoders of all the $L+t$ benefiting users (\emph{global CSIR}).

To demonstrate the structural origins of these CSI costs, we focus on a simple instance of the multi-server method in \cite{shariatpanahiMultiserverTransIT2016}, which acts as a proxy to other methods with similar feedback requirements.
\begin{ex}\label{ex:HighCSI}
	Let us consider the $L=2$-antenna MISO BC, with $K=4$ receiving users and cumulative cache size $t=2$. In this setting, the algorithm of \cite{shariatpanahiMultiserverTransIT2016} can treat $L+t = 4$ users at a time. Assuming that users $1,2,3,4$ request files $A,B,C,D$, respectively, each of the three transmissions of \cite{shariatpanahiMultiserverTransIT2016} takes the form\footnote{For sake of readability, in the examples provided throughout the document, we will omit the commas between numbers belonging to a set, such that, for example, the part of the file $A$ stored at the users in the set $\{3,4\}$ will be denoted by $A_{34}$.}
	\begin{equation}\label{eq:example1}
	\begin{aligned}
		\!\mathbf{x} = {}&\mathbf{h}^{\perp}_{4}(A_{23}\oplus B_{13}\oplus C_{12})+\mathbf{h}^{\perp}_{3}(A_{24}\oplus B_{14}\oplus D_{12}) + {}
		\\
		&\! +\mathbf{h}^{\perp}_{2}(A_{34}\oplus C_{14}\oplus D_{13})+\mathbf{h}^{\perp}_{1}(B_{34}\oplus C_{24}\oplus D_{23})
	\end{aligned}
	\end{equation}
where $\mathbf{h}^{\perp}_{k}$ denotes the precoder that is orthogonal to the channel of user $k$, and where
 $W_{ij}$ denotes the part of file $W\in\{A,B,C,D\}$ that is cached at users $i$ and $j$.
We can see that the transmitter must know all users' channel vectors, $\mathbf{h}_{k}, ~k\in\{1,2,3,4\}$, in order to form the four precoders.
In addition, each receiver must know the composite channel-precoder product for each precoder in order to be able to decode the desired subfile (e.g. receiver 1 must know $\mathbf{h}_{1}^{\dagger}\mathbf{h}_{1}^{\perp}$ as well as $\mathbf{h}_{1}^{\dagger}\mathbf{h}_{2}^{\perp}$, $\mathbf{h}_{1}^{\dagger}\mathbf{h}_{3}^{\perp}$ and $\mathbf{h}_{1}^{\dagger}\mathbf{h}_{4}^{\perp})$.
This implies a feedback cost equal to $L+t=4$ feedback-bearing users per transmission\footnote{In practical terms, this implies $L+t=4$ uplink training slots for CSIT acquisition and $L+t=4$ downlink training slots for global CSIR acquisition. We note that global CSIR acquisition can be performed by communicating each precoder to all users simultaneously, a process that is described in Appendix~\ref{sec:csi}.}.
\end{ex}

As we know (see for example~\cite{GrassmanTse2002,kobayashiCsiDoFcomparisonISIT2012}), such scaling feedback costs\footnote{In general we note that, in the context of Frequency Division Duplexing, the previously mentioned feedback results in a CSIT cost of $L+t$ feedback vectors. On the other hand, in a Time Division Duplexing environment, it leads to $L+t$ uplink training time slots for CSIT acquisition and an additional cost of $L+t$ downlink training time slots for global CSIR acquisition.
} can consume a significant portion of the coherence time, thus resulting in diminishing DoF gains.

\subsection{State of art}Motivated by this feedback bottleneck, different works on multi-antenna (and multi-transmitter) coded caching have sought to reduce CSI costs. However, in all known cases, any subsequent CSI reduction comes at the direct cost of substantially reduced DoF. For example, the works in \cite{ZFE:15,8007039} consider reduced quality CSIT, but yield a maximum DoF that is bounded close to $t+1$. Further, the works in \cite{nejadMixedCSITICC2017,zhang2017fundamental,piovanoGDoFMISO2019TransIT} consider delayed or reduced quality CSIT at the expense though of lower DoF performance, while the work in \cite{lampiris2017cache} considers only statistical CSI, but again achieves significantly lower DoF. Moreover, the work in \cite{ngo2018scalable} uses ACK/NACK type CSIT to ameliorate the issue of unequal channel strengths (cf.~\cite{lampirisUnevenChannelsIZS2020,joudehMixedTrafficArXiv2020}), yet achieving no multiplexing gains.
Similar results can be found in \cite{maddahCacheAidedICisit2015,roigCachesBothEndsICC2017,xuCooperativeTxRxISIT2016,hachemInterferenceNetsTransIT2018,senguptaCacheAidedWirelessCISS2016,zhangFundamentalCloudTransIT2019,cao2016fundamental,caoTreatingContentTransWireless2019}, in more decentralized scenarios that involve multiple cache-aided transmitters.

As a conclusion, both for the cache-aided MISO BC \cite{shariatpanahiMultiserverTransIT2016} as well as for its multi-transmitter equivalent~\cite{7857805}, the corresponding DoF $\mathcal{D}_{L}(t) = L+t$, has been known to require perfect feedback from all $L+t$ served users.

\subsection{Summary of contributions} The focus of this work is to establish and achieve the optimal relationship between feedback costs and DoF performance in multiple-antenna cache-aided settings.

As a consequence of our work, we now know that:
\begin{enumerate}
		\item The optimal DoF of the cache-aided MISO BC, under the assumptions of uncoded placement and one-shot linear schemes, takes the form
		\begin{equation}
			\mathcal{D}_{L}(t) =  L + t,
		\end{equation}
		{which tightens the previously known bound by a multiplicative factor of $2$.}
		{
		\item The optimal --- under the same assumptions --- DoF when feedback is limited to $C\ge 1$ participating users, takes the form
		\begin{equation}
			\mathcal{D}_{L}(t, C) = t + \min (L, C).
		\end{equation}
		}
		{
		\item Similarly, in the multi-transmitter scenario, with transmitter-side cumulative cache size $t_T$ and with each transmitter having $L_T$ antennas, the above optimal DoF performance takes the form
			\begin{equation}
				\mathcal{D}_{t_{T}\cdot L_{T}}(t, C) = \min ( t_T\cdot L_T, C)  + t.
			\end{equation}
		}
\end{enumerate}

The above are a direct outcome of a completely novel coded caching algorithm, which manages to achieve the optimal performance given any amount of available feedback. In particular, in the $L$-antenna MISO BC, or in the equivalent fully connected multi-antenna multi-transmitter setting with $t_{T} L_{T}= L$:
\begin{enumerate}
	\item The algorithm manages to achieve the optimal DoF
	\begin{equation}
		\mathcal{D}_{L}(t)= L+t
	\end{equation}
	and do so with a minimal feedback cost
	\begin{equation}
		C=L
	\end{equation}
	which substantially diminishes the previously known cost of $L+t$.
	\item The algorithm optimally degrades its DoF to
	\begin{equation}
		\mathcal{D}_{L}  ( t, C)= {C+ t}
	\end{equation}
when feedback is reduced to $C\in \{2,..., L-1\}$. This is an improvement over the state of art which, for the same DoF, would require a feedback cost of $C+t$.
The novelty of our scheme lies in the deviation from the traditional clique-based structure that most schemes are based on. Rather than requiring from each user to ``cache-out'' $t$ subfiles in a XOR as is commonly done, we are able to design transmissions that can benefit from a two-pronged approach: some users cache-out $t+L-1$ subfiles, and thus do not require the assistance of CSI-aided precoding, while others only cache-out $\frac{t}{L}$ subfiles but for that they rely on feedback. This allows our scheme to avoid the need to eventually ``steer-away'' subfiles from every active user, which had been the reason for the high feedback costs in all known prior designs.
\end{enumerate}

Finally an important contribution of this work can be found in the novel outer bound. This bound extends the effort in \cite{7857805} in two crucial ways.
\begin{itemize}
	\item A main contribution of the converse result is the incorporation of the limited feedback constraint. We integrate this new restriction by characterizing its impact on the number of users that we can serve simultaneously, which is obtained by exploiting the dimensionality of the linear system implicit in the multi-user transmission with constrained feedback.
	\item Apart from the contribution of being able to account for feedback-limited transmissions, we are able to improve the converse result of~\cite{7857805} by leveraging on the following insights: First, we exploit the symmetry of the configuration, which allows us to express the objective function of the optimization problem in terms only of the number of transmitters and the number of receivers that are caching each packet --- eliminating any dependence on the specific packet or node. Second, this symmetrization allows us to eventually produce an objective function that has monotonicity and convexity properties which in turn allow us to manipulate the solution to yield a tight bound. It might be worth noting that, as a result of this new approach, our converse also establishes exact optimality for a few subsequent works.
\end{itemize}

 \subsection{Notation} Symbols $\mathbb{N}, \mathbb{C}$ denote the sets of natural and complex numbers, respectively. For $n,k\in\mathbb{N},~n\geq k$, we denote the binomial coefficient with $\binom{n}{k}$, while $[k]$ denotes the set $\{1,2, ..., k\}$. For the bitwise-XOR operator we use $\oplus$. Greek lowercase letters are mainly reserved for sets. We further assume that all sets are ordered, and we use $| \cdot |$ to denote the cardinality of a set. Bold lowercase letters are reserved for vectors, while for some vector $\mathbf{h}$, comprised of $Q$ elements, we denote its elements as~$\mathbf{h}(q)$,~$q\in[Q]$, i.e., $[ \mathbf{h}(1), \mathbf{h}(2), ..., \mathbf{h}(Q)]\triangleq \mathbf{h}^{T}$. Bold uppercase letters are used for matrices, while for some matrix $\mathbf{H}$ we denote its $i$-th row, $j$-th column element as $\mathbf{H}(i,j)$.

\section{System Model}\label{se:sys_mod}

We consider the cache-aided MISO BC where an $L$-antenna transmitter serves $K$ single-antenna receiving cache-aided users. The distributed version of this setting, with multiple cache-aided transmitters, is discussed in Section~\ref{app:achiev_multiTX}.

In our setting, the transmitter has access to a library of $N\ge K$ files $\mathcal{F}=\{W^{(n)}\}_{n=1}^{N}$, of equal size. Each user has a cache that can fit a fraction $\gamma\in[0,1]$ of the library, and thus, collectively the users can store $t = K\gamma$ times over the entire library.
We assume that during delivery the users request their desired file simultaneously, and that each requested file is different. The users' file demand vector is denoted as $\mathbf{d} = \{d_{1}, ..., d_{K}\}$, implying that each user $k$ will request file $W^{(d_{k})}$.  In this setting, the received signal at user $k\in[K]$ takes the form
\begin{equation}
	y_{k}= \mathbf{h}_{k}^{\dagger} \mathbf{x}+w_{k},
\end{equation}
where $\mathbf{x}\in\mathbb{C}^{L \times 1}$ denotes the transmitted signal-vector from the $L$-antenna transmitter satisfying the power constraint $\mathbb{E}\left\{\|\mathbf{x}\|^{2}\right\}\le P$. In the above, $\mathbf{h}_{k}\in\mathbb{C}^{L\times 1}$ denotes the random-fading channel vector of user $k$, {which is assumed to be drawn from a continuous non-degenerate distribution}. This fading process is assumed to be statistically symmetric across users. Finally, the additive noise $w_k\thicksim\mathbb{C}\mathcal{N}(0,1)$ experienced at user $k$ is assumed to be Gaussian. The work focuses on the DoF performance, and thus the SNR is considered to be large. We also assume that the quality of CSIT is perfect, and we define the feedback amount required in each (packet) transmission as follows.

	\begin{definition}[Feedback Cost]\label{def_feedback_costs}
		A communication is said to induce \emph{feedback cost} $C$ if $C$ users need to communicate their CSI at the transmitter, and at the same time the transmitter needs to communicate information on $C$ precoders to the users.
	\end{definition}

\paragraph*{Structure of the paper} In Section~\ref{sec:main} we present the main results of this work and provide a preliminary example of the achievable scheme. Further, in Section~\ref{sec:SchemeDescription} we fully describe the achievable scheme for the single transmitter ($L$ antennas) case, and elaborate on the example of Section~\ref{sec:main}. In Section~\ref{app:achiev_multiTX} we extend the scheme to the multi-transmitter case. In Section~\ref{secConverse} we describe the proof of the converse result, while in Section~\ref{secConclusion} we provide general conclusions. The subsequent appendices include proofs, as well as a discussion on the CSIT and global CSIR feedback acquisition process that conveys the precoder information to the users.

\section{Main Results and an Example}\label{sec:main}
	We proceed with our main results, first by considering the single transmitter case (with $L$ transmit antennas), and later by extending the result to the general $K_T$-transmitter setting.
		We remind the reader that optimality is under the assumptions of one-shot linear schemes with uncoded cache placement, while we note that we directly omit the trivial bound $\mathcal{D}_{L} ( t, C) \le K$, and that we also do not consider the case of $C=0$ as this corresponds to the well known result in \cite{maddah2014fundamental}. We additionally recall that the setting asks that each of the $K$ receiving users is equipped with an identically-sized cache of normalized size $\gamma$, thus corresponding to a cumulative receiver-side cache size of $t = K\gamma$.

	\begin{theorem}\label{theoSingleTXresult}
		In the $K$-user cache-aided MISO BC with $L$ transmit antennas, cumulative cache size $t$, and feedback cost $C$, the optimal DoF is
		\begin{equation}
			\mathcal{D}_{L}(t, C)=t+\min(L, C).
		\end{equation}
	\end{theorem}
	\begin{proof}
		The achievability part is constructive and is described in Section~\ref{sec:SchemeDescription}, while the converse is proved in Section~\ref{secConverse}.
	\end{proof}

	Let us consider now the more general setting where the $L$-antenna transmitter is substituted by $K_{T}$ cache-enabled transmitters. Each transmitter is equipped with $L_{T}$ transmit antennas, and is able to store a fraction $\gamma_{T} \in [1/K_T,1]$ of the library, inducing a cumulative cache size of $t_{T}\triangleq K_{T} \gamma_T$.
	\begin{theorem}\label{theoMultiTXresult}
		In the $K_{T}$-transmitter wireless network, where each transmitter is equipped with $L_{T}$ transmit antennas, with transmitter-side cumulative cache size $t_{T}$, receiver-side cumulative cache size $t$, and feedback cost $C$, the optimal DoF takes the form
		\begin{equation}
			\mathcal{D}_{L_{T} t_{T}}(t, C)=t+\min(L_{T}t_{T}, C).
		\end{equation}
	\end{theorem}
	\begin{proof}
		The achievability part of the proof is described in Section~\ref{app:achiev_multiTX}, while the converse is described in Section~\ref{secConverse}.
	\end{proof}	
	\begin{remark}\label{re:remark_equate}
		Comparing Theorem~\ref{theoSingleTXresult} with Theorem~\ref{theoMultiTXresult}, we can see that the cache-aided MISO BC and its multi-transmitter equivalent (corresponding to $L_{T}t_{T} = L$) are akin not only in terms of DoF performance, but also in terms of the CSIT required to achieve this performance. Their behavior is the same, irrespective of the amount $C\ge 1$ of available feedback.
	\end{remark}
The following corollary establishes the \emph{exact} optimal DoF performance of the considered multi-antenna settings.
\begin{corollary}
	The optimal DoF of the $L$-antenna MISO BC with $K$ users and cumulative cache size $t$, takes the form
	\begin{equation}
		\mathcal{D}_{L}(t) = L+t.
	\end{equation}
\end{corollary}
	\begin{remark}\label{needed_csit}
			The DoF performance $\mathcal{D}_{L}(t, C)=L+t$ can be achieved by knowing the CSIT of only $C=L$ users at each transmission.
	\end{remark}
	\begin{remark}
		As Theorems \ref{theoSingleTXresult},~\ref{theoMultiTXresult} demonstrate, in order to achieve the maximum one-shot linear DoF $\mathcal{D}_{L}(t, C) = L+t$ in the multi-antenna case ($L>1$), the condition $C\ge L$ is both sufficient and necessary.
	\end{remark}
	
	\begin{remark} In several scenarios such as in~\cite{piovanoCCparallelChannelsTComm2020,naderializadehCellularNetsTComm2019}, the best known bounds --- which are built on the converse proof of~\cite{7857805} --- endure a multiplicative gap to the optimal performance.
Our converse proof improves the converse in\cite{7857805} by tightening the lower bound of the solution to the linear program proposed in\cite{7857805}. Consequently, our converse also closes the multiplicative gap of such subsequent works.
For example, it follows directly from the results derived here that the achievable DoF presented in\cite{piovanoCCparallelChannelsTComm2020} for a cache-aided interference network with heterogeneous parallel channels and centralized cache placement is in fact exactly optimal. Similarly, the achievable DoF in~\cite{naderializadehCellularNetsTComm2019} for cache-aided cellular networks again turns out to be exactly optimal.
	\end{remark}

\subsection*{Intuition and an example of the scheme}

Revisiting the previous optimal multi-antenna coded caching algorithms (cf.~\cite{shariatpanahiMultiserverTransIT2016,7857805,shariatpanahi2017multi,tolli2017multi}) --- which as we noted, require CSIT from all $L+t$ ``active'' users --- we remark that the main premise of these designs is that each transmitted subfile can be cached-out by some $t$ users (as in the algorithm of \cite{maddah2014fundamental}), and at the same time it can be zero-forced at some other $L-1$ users. This, in turn, allows each of the $L+t$ active users to receive its desired subfile, free of interference. This design, while achieving the maximum DoF, incurs very high CSIT costs. Notably, these costs are associated with the need to eventually ``steer-away'' subfiles from every active user.

The idea that we follow is different. In order to reduce the amount of CSIT to be from only $L$ feedback-aided users, while retaining the full DoF performance, it follows that the $t$ users whose CSI is unknown (hereon referred to as the set $\pi$) will need to cache-out a total of $t+L-1$ subfiles each. On the other hand, the $L$ users whose CSI is known (hereon referred to as the set $\lambda$) will be assisted by precoding, and can thus more easily receive their desired subfile. Hence, the main design challenge is to transmit together subfiles that can be decoded by each user of set $\pi$.

We proceed with a preliminary description of the proposed algorithm.

\subsubsection*{Algorithm overview}

We first note that the cache placement draws directly from \cite{maddah2014fundamental}, both in terms of file partition, as well as in terms of storing of subfiles in the users' caches.

On the other hand, the XOR generation method will be fundamentally different. The first step is to construct XORs composed of $\frac{t}{L}+1$ subfiles and to compose each transmit-vector with $L$ such XORs. This allows each transmission to communicate $L+t$ different subfiles aimed at serving, simultaneously, a set of $L+t$ users. As discussed above, each such set of $L+t$ active users is partitioned into two sets; the first set, $\lambda$, consists of the $L$ users that are assisted by precoding. The second set, $\pi$, has $t$ users who are not assisted by precoding and who must compensate with their caches. The vector of XORs will be multiplied by $\mathbf{H}_{\lambda}^{-1}$ which represents the normalized inverse of the channel matrix between the transmitter and the users in set $\lambda$.

We will see that the above design guarantees that, during the decoding process, each of the users in $\lambda$ only receives one of the XORs (because the rest will be nulled-out by the precoder), while the remaining $t$ users, i.e., those in $\pi$, receive a linear combination of all $L$ XORs. Hence, this means that each user in $\lambda$ needs to cache out $\frac{t}{L}$ subfiles in order to decode its desired subfile, while the users in $\pi$ need to cache out $t+L-1$ subfiles, i.e., all but one subfiles.

\subsubsection*{Algorithm demonstration through an example}
Next, we will demonstrate a single transmission of our algorithm by considering the setting of Example~\ref{ex:HighCSI}. The goal is to achieve the same performance as before (delivery to four users at a time) while using CSIT from only two users at a time. The example in its entirety can be found in Section~\ref{sec:DoFperfCalculation}, Example~\ref{ex:fullSmall}.

\begin{ex}\label{ex:firstSmallScheme}
In the same MISO BC setting of Example~\ref{ex:HighCSI} with $L=2$ transmit antennas, $K=4$ users, and cumulative cache size $t=2$, one transmitted vector of the proposed algorithm takes the form\footnote{The reader is warned that there is a small notational discrepancy between the subfile indices of this example and the formal notation. In this example we have kept the notation as simple as possible in order to more easily provide a basic intuition on the structure of the scheme.}
\begin{align}
	\mathbf{x}=\mathbf{h}_{2}^{\perp}( A_{34}\oplus C_{14})+\mathbf{h}_{1}^{\perp}(B_{34}\oplus D_{23}),
\end{align}
where $\mathbf{h}_{k}^{\perp}$, $k\in\{1,2\}$, denotes the precoder-vector designed to be orthogonal to the channel of user $k$, files $A$, $B$, $C$, and $D$ are requested by users $1,2,3,$ and $4$, respectively, and where $W_{ij}$, $W\in\{A, B,C,D\}$, represents the subfile of $W$ that can be found in the caches of users~$i$ and~$j$.

Assuming that user $k$ receives $y_{k},~k\in [4]$,  the message at each user takes the form
\begin{align}
	\begin{bmatrix}
		y_{1} \\ y_{2} \\ y_{3} \\ y_{4}
	\end{bmatrix} \! = \! \begin{bmatrix}
		\mathbf{h}_{1}^{\dagger}(\mathbf{h}_{2}^{\perp} A_{34}\oplus C_{14}+\mathbf{h}_{1}^{\perp}B_{34}\oplus D_{23})\\
		\mathbf{h}_{2}^{\dagger}(\mathbf{h}_{2}^{\perp} A_{34}\oplus C_{14}+\mathbf{h}_{1}^{\perp}B_{34}\oplus D_{23})\\				\mathbf{h}_{3}^{\dagger}(\mathbf{h}_{2}^{\perp} A_{34}\oplus C_{14}+\mathbf{h}_{1}^{\perp}B_{34}\oplus D_{23})\\
		\mathbf{h}_{4}^{\dagger}(\mathbf{h}_{2}^{\perp} A_{34}\oplus C_{14}+\mathbf{h}_{1}^{\perp}B_{34}\oplus D_{23})
	\end{bmatrix}
	\nonumber
	\\
	\!=\!
	\begin{bmatrix}
		A_{34}\oplus C_{14}\\
		B_{34}\oplus D_{23}\\			
		\mathbf{h}_{3}^{\dagger}(\mathbf{h}_{2}^{\perp} A_{34}\oplus C_{14}+\mathbf{h}_{1}^{\perp}B_{34}\oplus D_{23})\\
		\mathbf{h}_{4}^{\dagger}(\mathbf{h}_{2}^{\perp} A_{34}\oplus C_{14}+\mathbf{h}_{1}^{\perp}B_{34}\oplus D_{23})
	\end{bmatrix}
\end{align}
where we have ignored noise for simplicity.

Hence, we see that users $1$ and $2$ only receive the first and second XOR, respectively, due to the precoder design. This means that each of these two users can decode its desired subfiles, $A_{34}$ and $B_{34}$, respectively, by caching-out the unwanted subfiles $C_{14}$ and $D_{23}$, respectively.

On the other hand, looking at the decoding process for users 3 and 4, we see that user $3$ can cache-out subfiles $A_{34}$, $B_{34}$, and $D_{23}$ in order to decode the desired $C_{14}$. Similarly, user $4$ can cache-out subfiles $A_{34}$, $B_{34}$, and $C_{14}$ to decode the desired subfile $D_{23}$. In order to achieve this, users $3$ and $4$ need to employ their cached content, but they also need some CSI knowledge: user 3 needs products $\mathbf{h}_{3}^{\dagger}\mathbf{h}_{2}^{\perp}$ and $\mathbf{h}_{3}^{\dagger}\mathbf{h}_{1}^{\perp}$, while user 4 needs $\mathbf{h}_{4}^{\dagger}\mathbf{h}_{2}^{\perp}$ and $\mathbf{h}_{4}^{\dagger}\mathbf{h}_{1}^{\perp}$. This can be handled with the broadcasting of information for only two precoders. The reader is referred to Appendix~\ref{sec:csi} for an exposition of how the feedback acquisition here requires only $L=2$ training slots, which is simply because information on a precoder can be broadcasted in a single shot, irrespective of how many users it is broadcasted to.
\end{ex}

\section{Description of the Scheme \label{sec:SchemeDescription}}
We proceed to present the scheme's cache-placement and content-delivery phases. We focus on the single transmitter MISO BC setting, while the multi-transmitter scenario is presented in Section~\ref{app:achiev_multiTX}.
Furthermore, we also assume in the following that $C= L$, noting that the extension of the scheme to the case $C<L$ is trivial and it can be achieved by simply ``shutting down'' $L-C$ antennas. The scheme is described for the case where $\frac{t}{L}\in\mathbb{N}$, while the remaining cases can be achieved using memory sharing and, as shown in \cite{lampiris2018adding}, would incur a small DoF reduction\footnote{Efforts subsequent to our work \cite{lampiris2018full,lampirisBridgingAsilomar2019} have addressed this memory sharing issue through a new design that is able to retain the same desirable DoF and feedback cost without being constrained by the value of $L$.}.

Communication happens in two phases, namely the placement and the delivery phases. The placement phase is responsible for populating the caches of the users with content, while the delivery phase is responsible for communicating to the users their desired files. Further, we assume that each transmission occupies multiple coherence periods.

\paragraph*{Precoder design}For some set $\lambda\subset[K]$ of $|\lambda|=L$ users, we denote with $\mathbf{H}_{\lambda}^{-1}$ the normalized inverse of the $L\times L$ channel matrix $\mathbf{H}_{\lambda}$ corresponding to the channel between the transmitter and the $L$ users of set $\lambda$. Further, the $\ell$-th column of $\mathbf{H}_{\lambda}^{-1}$, $\ell\in [L]$, is denoted by $\mathbf{h}^{\perp}_{\lambda \setminus \lambda(\ell)}$ and describes a vector that is orthogonal to the channels of the users of set $\lambda \setminus \lambda(\ell)$. Hence, for some arbitrary user $k\in[K]$, it holds that
\begin{align}\label{eqPrecoderDesign}
	\mathbf{h}_{k}^{\dagger}\cdot\mathbf{h}^{\perp}_{\lambda \setminus \lambda(\ell)}
	\begin{cases}
		= 0, & \text{if } k \in \lambda \setminus \lambda(\ell)\\
		\neq	0, & \text{else. }
	\end{cases}
\end{align}

\subsection{Placement phase}
The placement phase is executed without knowledge of the number of transmit antennas, and without knowledge of CSI. The placement follows the original scheme in \cite{maddah2014fundamental} where each file $W^{(n)},n\in[N]$, is initially split into $\binom{K}{t}$ subfiles
\begin{equation}
	W^{(n)} \to \left\{ W_{\tau}^{(n)} , \ \ \tau \subset[K],\ |\tau|= t\right\},
\end{equation}
each indexed by a $t$-length set $\tau \subset[K]$, such that the cache of user $k\in[K]$ takes the form
\begin{equation}\label{eq:initialSubpacketizaion}
	\mathcal{Z}_{k}=\left\{W_{\tau}^{(n)}: \forall \tau\ni k,  |\tau| = t, \forall n\in[N]\right\}.
\end{equation}

\subsection{Delivery phase}

This phase begins with the request from each user of a single file from the library. To satisfy these demands, the transmitter selects a subset of $L+t$ users for each transmission slot. Specifically, these users are divided into set $\lambda\subset[K], \ |\lambda|= L$, who provide CSI, and set $\pi\subset[K]\setminus\lambda,\ |\pi| = t$, who need not provide CSI.

Upon notification of the requests $\{W^{(d_k)}, \ k\in[K]\}$, and after the number of antennas is revealed to be $L$,
each requested subfile $W_{\tau}^{(d_{k})}$ is further split twice as follows:
\begin{align}\label{eq:sigmaSubpacketization}
	W_{\tau}^{(d_{k})} \to&\{ W^{(d_{k})}_{\sigma,\tau},~
	\sigma\subseteq[K]\setminus(\tau\cup\{k\}), |\sigma|=L-1\} \\ \label{eq:sigmaSubpacketizationPhi}
	W_{\sigma,\tau}^{(d_{k})}\to&\{ W^{r,(d_{k})}_{\sigma, \tau}	, ~  r\in[L+t]\}.
\end{align}

In the following we describe how, for every transmission, the transmitter first creates a vector of $L$ XORs, and then precodes each XOR with the appropriate precoder.
\paragraph{Individual XOR design}
As previously mentioned, each transmitted XOR has $t/L+1$ recipients, which we refer to as set $\mu$. We recall that each subfile is cached at $t$ receivers, and we consider the set $\nu$ to be the set of $t- (t/L+1) = t\frac{L-1}{L}$ users who have cached the set of files intended for users in set $\mu$. In particular, these two sets $\mu,\nu\subset[K]$, are disjoint ($\mu\cap\nu=\emptyset$), and their sizes are $|\mu|=\frac{t}{L}+1$ and ~$|\nu|=t\frac{L-1}{L}$ respectively. We also consider a set $\sigma\subseteq\Big([K]\setminus(\mu\cup\nu)\Big), ~ |\sigma|=L-1$, which will be later chosen more carefully. With these in place, we construct XOR\footnote{In a small abuse of notation, we will henceforth refer to the segments of the original subfiles again as subfiles. We also note that, for clarity of exposition and to avoid many indices, index $r$ of~\eqref{eq:sigmaSubpacketizationPhi} will henceforth be suppressed, thus any $ W^{r,(d_{k})}_{\sigma, \tau}$ will be denoted as $W_{\sigma,\tau}^{(d_{k})}$ unless $r$ is explicitly needed.}
\begin{equation}\label{eq:XORexplanation}
	X_{\mu}^{\nu, \sigma}=\bigoplus_{k\in\mu} W^{(d_{k})}_{\sigma,(\nu\cup \mu)\setminus \{k\}}
\end{equation}
which consists of $\frac{t}{L}+1$ subfiles, where
\begin{itemize}
\item each subfile in the XOR is requested by one user in $\mu$, and where
\item all subfiles of the XOR are known by all users in $\nu$.
\end{itemize}
The set $(\nu\cup \mu)\setminus \{k\}$ plays the role of $\tau$ from the placement phase, as it describes the set of users that have this subfile (labeled by $\tau$) in their cache, while set $\sigma$ is a selected subset of $L-1$ users from set $\lambda$.

\begin{ex}
Let us consider the MISO BC with $L=2$ transmit antennas, $K\ge 6$ users and cumulative cache size $t = 4$.
Let the aforementioned sets be $\mu = \{1,2,3\}$, $\nu = \{4,5\}$, and consider some arbitrary $\sigma\subseteq [K]\setminus\{1,2,3,4,5\}$, $|\sigma|=1$. Then, the XOR of~\eqref{eq:XORexplanation} takes the form
\begin{align}
	X_{123}^{45,\sigma}=W_{\sigma,{\scriptsize \underbrace{ 2345}_\tau} }^{(d_{1})}\oplus W_{\sigma,1345}^{(d_{2})}\oplus W_{\sigma,1245}^{(d_{3})}.
\end{align}
As we have described before, this XOR delivers subfiles desired by all the users of set $\mu$, while each element of the XOR is cached at all users of set $\nu$.
It is easy to see that users $1,2$, and $3$ work in the traditional way to cache out the interfering subfiles in order to get their own desired subfile, such that for example user~1 caches out $W_{\sigma,1345}^{(d_{2})}\oplus W_{\sigma,1245}^{(d_{3})}$ to get its own $W_{\sigma,2345}^{(d_{1})}$.  In turn, users $4$ and $5$ are fully protected against this entire undesired XOR because they have cached all 3 subfiles of this XOR.
As a quick verification, we see that each index $\tau$ has size $|\tau| = t = 4$, which adheres to the available cache-size constraint as each file can be stored at exactly $t=4$ receivers.
\end{ex}

\paragraph{Design of vector of XORs}
\begin{algorithm}[h]\caption{Delivery Phase}\label{alg:Delivery}
 \For{$\lambda\subset[K], |\lambda|=L$ (precoded users in $\lambda$)}{
	Calculate $\mathbf{H}^{-1}_{\lambda}$ \\
	\For{$\pi\subseteq \left([K]\setminus \lambda\right),\  |\pi|=t$}{
		Break $\pi$ into some $\phi_{i}\ i \in [L]:\ |\phi_{i}|=\frac{t}{L},$ $\bigcup_{i\in[L]} \phi_{i}=\pi, \   \phi_{i}\cap \phi_{j}=\emptyset,\forall i,j\in [L]$\\
		\For{$s\in \{0,1,...,L-1\}$}{
			$v_i=((s+i-1)\mod L)+1, i\in[L]$\\
			Transmit
		\begin{align}\label{eq:transmission}
		\mathbf{x}_{\lambda,\pi}^{s}=\mathlarger{\mathbf{H}}_{\lambda}^{-1}\cdot
			\begingroup
			\renewcommand*{\arraystretch}{2}
				\begin{bmatrix}
					\mathlarger{\mathlarger{X}}_{\lambda(1)\cup \phi_{v_{1}}}^{\pi\setminus \phi_{v_{1}},\lambda\setminus\lambda(1)}\\
					\mathlarger{\mathlarger{X}}_{\lambda(2)\cup \phi_{v_{2}}}^{\pi\setminus \phi_{v_{2}},\lambda\setminus\lambda(2)}\\
					\vdots\\
					\mathlarger{\mathlarger{X}}_{\lambda(L)\cup \phi_{v_{L}}}^{\pi\setminus \phi_{v_{L}},\lambda\setminus\lambda(L)}
				\end{bmatrix}
				.
			\endgroup
    		\end{align}
		
		}
	}
}
\end{algorithm}

Equipped with the design of each individual XOR, the goal is to select $L$ such XORs in order to communicate them in a single transmission period. Algorithm~\ref{alg:Delivery} forms a set of $L+t$ users and a set of $L$ distinct such XORs to serve them with. Specifically, the steps that are followed are described below.
\begin{itemize}
\item In Step 1, a set $\lambda$ of $L$ users is chosen.
\item In Step 2, a (ZF-type) precoder $\mathbf{H}_{\lambda}^{-1}$ is designed to spatially separate the $L$ users in $\lambda$.
\item In Step 3, another set $\pi \subseteq [K]\setminus \lambda$ of $t$ users is selected from the remaining users.
\end{itemize}
To construct the $L$ XORs and to properly place them in the vector, the following steps take place.
\begin{itemize}
\item In Step 4, set $\pi$ of $t$ users is arbitrarily partitioned into $L$ non-overlapping sets $\phi_{i},\ i\in [L]$, each having $\frac{t}{L}$ users.

\item Steps $5$ and $6$ are responsible for forming the $L$ different sets $\mu$ (cf. \eqref{eq:XORexplanation}), where each such set $\mu$ consists of $\frac{t}{L}+1$ users. Specifically, in every iteration of Step~$5$, the algorithm associates a user from set $\lambda$ with some set $\phi_{v_{i}}$, in order to form set $\mu$ and such that after $L$ iterations each user from $\lambda$ would be associated with every set $\phi_{v_i}$. For example, when $s=0$, the first XOR of the vector will be intended for users in set $\{\lambda(1)\}\cup\phi_{1}$ (while completely known by all users in $\pi\setminus\phi_{1}$), the second XOR will be intended for the users in the set $\{\lambda(2)\}\cup\phi_{2}$ (while completely known by all users in $\pi\setminus\phi_{2}$), and so on. Further, when $s=1$ the first XOR will be intended for users in $\{\lambda(1)\}\cup\phi_{2}$ (while completely known by all users in $\pi\setminus\phi_{2}$), the second XOR will be for users in $\{\lambda(2)\}\cup\phi_{3}$ (while completely known by all users in $\pi\setminus\phi_{3}$), and so on. In particular, Step $5$ (and the operation in Step $6$, as shown in~Algorithm~\ref{alg:Delivery}) allows us to iterate over all sets $\phi_{i}$, associating every time a distinct set $\phi_{i}$ to a distinct user from group $\lambda$, until all users from set $\lambda$ have been associated with all sets $\phi_{i}$. The verification that this association does not leave behind any subfiles is performed later on in this section.
\item Then, in the last step (Step $7$), the vector of the $L$ XORs
is transmitted after being precoded by matrix $\mathbf{H}_{\lambda}^{-1}$.

\end{itemize}

\paragraph{Decoding at the users}
By the very nature of the XOR design, as seen in~\eqref{eq:XORexplanation}), the vector of XORs we constructed in~\eqref{eq:transmission} guarantees that the users in $\lambda$ can decode the single XOR that they receive (recall that for such users, all other XORs are steered away due to ZF precoding) and can thus subsequently proceed to decode their own file through the use of their cached content. Further, the design guarantees that each user in $\pi$ has cached all subfiles that are found in the entire vector, apart from its desired subfile. Benefitting from their receiver-side CSI (see Appendix~\ref{sec:csi}), the users of set $\pi$ are provided with all the necessary CSI estimates, which allows for the decoding of the linear combination of the transmitted vector.

To see the above more clearly, let us look at the signal received, and the subsequent decoding process at some of the users.

For some user $\ell\in\lambda$, the decoding process is simple. The received message takes the form
\begin{align*}
	y_{\ell}&=\mathbf{h}^{\dagger}_{\ell} \mathbf{H}_{\lambda}^{-1}
\begingroup
			\renewcommand*{\arraystretch}{2}
				\begin{bmatrix}
					\mathlarger{\mathlarger{X}}_{\lambda(1)\cup \phi_{v_{1}}}^{\pi\setminus \phi_{v_{1}},\lambda\setminus\lambda(1)}\\
					\mathlarger{\mathlarger{X}}_{\lambda(2)\cup \phi_{v_{2}}}^{\pi\setminus \phi_{v_{2}},\lambda\setminus\lambda(2)}\\
					\vdots\\
					\mathlarger{\mathlarger{X}}_{\lambda(L)\cup \phi_{v_{L}}}^{\pi\setminus \phi_{v_{L}},\lambda\setminus\lambda(L)}
				\end{bmatrix}
			\endgroup
		=\mathlarger{\mathlarger{X}}_{\{\ell\}\cup \phi_{v_{k}}}^{\pi\setminus \phi_{v_{k}},\lambda\setminus\{\ell\}} \label{eqRemainingXOR1}
		,
\end{align*}
where $\phi_{v_{k}}$, $k\in[L]$, represents the subset of $\pi$, of size $|\phi_{v_{k}}|=\frac{t}{L}$, associated with $\ell$ (Step $5$ of Algorithm~\ref{alg:Delivery}). The selected precoders allow user $\ell$ to receive only one of the XORs (cf.~\eqref{eqRemainingXOR1}). Due to the design of this remaining XOR (see~\eqref{eq:XORexplanation}), all but one subfiles have been cached by user $\ell$, and thus the user can decode its desired subfile.

On the other hand, the decoding process at some user in set $\pi$ requires, also, access to CSI. The received message at user $p\in\pi$ takes the form
\begin{align}
	y_{p}&=\mathbf{h}^{\dagger}_{p} \mathbf{H}_{\lambda}^{-1}
\begingroup
			\renewcommand*{\arraystretch}{2}
				\begin{bmatrix}
					\mathlarger{\mathlarger{X}}_{\lambda(1)\cup \phi_{v_{1}}}^{\pi\setminus \phi_{v_{1}},\lambda\setminus\lambda(1)}\\
					\mathlarger{\mathlarger{X}}_{\lambda(2)\cup \phi_{v_{2}}}^{\pi\setminus \phi_{v_{2}},\lambda\setminus\lambda(2)}\\
					\vdots\\
					\mathlarger{\mathlarger{X}}_{\lambda(L)\cup \phi_{v_{L}}}^{\pi\setminus \phi_{v_{L}},\lambda\setminus\lambda(L)}
				\end{bmatrix}
			\endgroup
\\ &
		=\sum_{j=1}^{L} \mathbf{h}^{\dagger}_{p}\mathbf{h}^{\perp}_{\lambda\setminus\lambda(j)}	\mathlarger{\mathlarger{X}}_{\lambda(j)\cup \phi_{v_{j}}}^{\pi\setminus \phi_{v_{j}},\lambda\setminus\lambda(j)} \label{eq:decodingProcessInPi}
		.
\end{align}

First, we observe that, due to the process described in Appendix~\ref{sec:csi}, user~$p$ has estimated all products $\mathbf{h}^{\dagger}_{p}\mathbf{h}^{\perp}_{\lambda\setminus\{\ell\}}, ~\forall \ell\in\lambda$, that appear in \eqref{eq:decodingProcessInPi}. Then, by taking account of the fact that $\phi_{v_{i}}\cap\phi_{v_{j}} =\emptyset$ if $i\ne j$, we can see that user~$p$ belongs to one of the sets $\phi_{v_{j}}\subset\pi$. This means that user~$p$ has stored the content of all but one XORs (see~\eqref{eq:XORexplanation}) and can thus remove them from~\eqref{eq:decodingProcessInPi}. By removing the $L-1$ known XORs, the remaining message at user~$p$ is
\begin{equation}
	\mathbf{h}^{\dagger}_{p}\mathbf{h}^{\perp}_{\lambda\setminus\lambda(j)}	\mathlarger{\mathlarger{X}}_{\lambda(j)\cup \phi_{v_{j}}}^{\pi\setminus \phi_{v_{j}},\lambda\setminus\lambda(j)}
\end{equation}
where $\phi_{j} \ni p$. Due to its structure (cf.~\eqref{eq:XORexplanation}), the XOR can be successfully used by user $p$ to decode its own desired message.

\subsection{Evaluating the scheme's performance}\label{sec:DoFperfCalculation}

In order to calculate the achievable DoF of the proposed scheme, we begin by showing that each desired subfile of set $\{W^{r,(d_{k})}_{\sigma, \tau}\}_{r=1}^{L+t}$ is transmitted exactly once. Since each such collection of subfiles has the same sub-indices, it follows that there is no need to distinguish between them, as long as each appears exactly once.

\paragraph{Each desired subfile is transmitted exactly once}

For any arbitrary subfile $W_{\sigma,\tau}^{(d_{k})}$, the labeling $(\sigma, \tau, k)$ defines the set of active users $\lambda \cup \pi = \sigma\cup \tau \cup \{k\}$. Let us recall that $\lambda\cap \pi = \emptyset$, $\sigma\cap \tau = \emptyset$, that $\sigma\subset \lambda$, and that $|\sigma| = L-1$, $|\lambda| = L$, $|\pi| = |\tau| = t$.
For our fixed $\sigma,\tau,k$, let us consider the two complementary cases; case i) $k\in\lambda$, and case ii) $k\notin \lambda$.

In case i), $\lambda= \sigma\cup \{k\}$, {since $\tau\cap \lambda = \emptyset$.} 
Moreover,
\begin{align*}
	\pi = (\sigma\cup \tau \cup \{ k \}) \setminus \lambda = \tau
\end{align*}
means that a fixed $(\sigma, \tau, k)$ corresponds to a single $(\lambda,\pi)$. For any fixed $(\lambda,\pi)$ in Algorithm~\ref{alg:Delivery}, Step $5$ iterates $L$ times, thus identifying $L$ specific component subfiles which are defined by the same $(\sigma, \tau, k)$, and thus can be differentiated by $L$ different $r \in [t+L]$; these $L$ component subfiles of $W_{\sigma,\tau}^{(d_{k})}$ will appear in transmissions
$\mathbf{x}_{\lambda,\pi}^{s},~s=0,1,\dots,L-1$.

In case ii), the fact that $k\notin \lambda$ implies that for a given $(\sigma,\tau,k)$ (which also defines the set of active users) there can be $t$ different sets $\lambda$ which take the form
\begin{align*}
	\lambda = \sigma \cup \{ \tau(i) \}, ~~ i\in [t].
\end{align*}
This means that any fixed triplet $(\sigma,\tau, k)$ corresponds to $t$ different possible sets $\lambda$. Since for a fixed $(\sigma,\tau, k)$, the union of $\lambda\cup \pi$ is fixed, we can conclude that each fixed $(\sigma,\tau, k)$ is associated to $t$ different pairs $(\lambda,\pi)$.

Now, having chosen a specific pair $(\lambda,\pi)$, where we remind that $k\in\pi$, we can see from Step 5 of Algorithm \ref{alg:Delivery} that user $k$ should belong to exactly one set $\phi_{v_{i}},i\in[L]$. Let that set be $\phi_{v_{j}}$. This means that from all $L$ transmissions of Step 5, a component subfile of the form $W^{(d_{k})}_{\sigma,\tau}$ will be transmitted in exactly one transmission, and in particular, in the single transmission which includes XOR
\begin{equation*}
{\mathlarger{X}}_{\tau(i)\cup \phi_{v_{j}}}^{\pi\setminus \phi_{v_{j}},\sigma}.
\end{equation*}
In total, for all the different $(\lambda,\pi)$ sets, subfile $W^{(d_{k})}_{\sigma,\tau}$ will be transmitted $L+t $ times.
Finally, since we showed that an arbitrary subfile, $W^{(d_{k})}_{\sigma,\tau}$, will be transmitted exactly $L+t$ times, this implies that all subfiles of interest will be transmitted once we go over all possible $\lambda, \pi$ sets.

\paragraph{DoF calculation}
The resulting DoF can now easily be seen to be $L+t$, simply because each transmission includes $L+t$ different subfiles, and because each file was indeed transmitted exactly once. A quick verification, accounting for the subpacketization
\begin{align*}
	\mathcal{S}_{L}=\binom{K}{t}\binom{K-t-1}{L-1}(L+t),
\end{align*}
and accounting for the number of iterations in each step, tells us that the worst-case delivery time takes the form
\begin{align}\label{eq:timeCalc}
\mathcal{T}_{L}(t)=\frac{\overbrace{\binom{K}{L}}^{\text{Step }1}\overbrace{\binom{K-L}{t}}^{\text{Step } 3}\cdot \overbrace{L}^{\text{Step }5}}{\binom{K}{t}\binom{K-t-1}{L-1}(L+t)} =\!\frac{K-t}{L+t},
\end{align}
which in turn directly implies a DoF of
\begin{align*}
	\mathcal{D}_{L}(t) = \frac{K(1-\gamma)}{\mathcal{T}_{L}(t) }= L+t
\end{align*}
which is achieved with CSI from only $C=L$ users per transmission. \qed

To illustrate the above algorithm, we proceed to present the delivery phase for the setting of Example~\ref{ex:firstSmallScheme}.

\begin{ex}[Example of scheme]\label{ex:fullSmall}
Consider a transmitter with $L=2$ antennas, serving $K=4$ users with cumulative cache size $t=2$.
Each file is split into
\begin{align*}
	\mathcal{S}_{L}=\overbrace{(t+L)}^{r}\overbrace{\binom{K-t-1}{L-1}}^{\sigma}\overbrace{\binom{K}{t}}^{\tau}=24
\end{align*}
subfiles. The $\binom{K}{L} \binom{K-L}{t} L = 12 $ transmissions that satisfy all the users' requests are
\begin{align*}
	&\mathbf{x}_{12,34}^{1}\!=\!\mathbf{H}_{12}^{-1}
	\begin{bmatrix}
		A_{2,34}^{(1)}\!\oplus\! C_{2,14}^{(1)}\\
		B_{1,34}^{(1)}\!\oplus\! D_{1,23}^{(1)}
	\end{bmatrix}
	,
	\mathbf{x}_{12,34}^{2}\!=\!\mathbf{H}_{12}^{-1}
	\begin{bmatrix}
		A_{2,34}^{(2)}\!\oplus\! D_{2,13}^{(1)}\\
		B_{1,34}^{(2)}\!\oplus\! C_{1,24}^{(1)}
	\end{bmatrix}
	\\
	&\mathbf{x}_{34,12}^{1}\!=\!\mathbf{H}_{34}^{-1}
	\begin{bmatrix}
		B_{4,13}^{(1)}\!\oplus\! C_{4,12}^{(1)}\\
		A_{3,24}^{(1)}\!\oplus\! D_{3,12}^{(1)}
	\end{bmatrix}
	,
	\mathbf{x}_{34,12}^{2}\!=\!\mathbf{H}_{34}^{-1}
	\begin{bmatrix}
		A_{4,23}^{(1)}\!\oplus\! C_{4,12}^{(2)}\\
		B_{3,14}^{(1)}\!\oplus\! D_{3,12}^{(2)}
	\end{bmatrix}
	\\
	&\mathbf{x}_{24,13}^{1}\!=\!\mathbf{H}_{24}^{-1}
	\begin{bmatrix}
		A_{4,23}^{(2)}\!\oplus\! B_{4,13}^{(2)}\\
		C_{2,14}^{(2)}\!\oplus\! D_{2,13}^{(2)}
	\end{bmatrix}
	,
	\mathbf{x}_{24,13}^{2}\!=\!\mathbf{H}_{24}^{-1}
	\begin{bmatrix}
		B_{4,13}^{(3)}\!\oplus\! C_{4,12}^{(3)}\\
		A_{2,34}^{(3)}\!\oplus\! D_{2,13}^{(3)}
	\end{bmatrix}
	\\
	&\mathbf{x}_{13,24}^{1}\!=\!\mathbf{H}_{13}^{-1}
	\begin{bmatrix}
		A_{3,24}^{(2)}\!\oplus\! B_{3,14}^{(2)}\\
		C_{1,24}^{(2)}\!\oplus\! D_{1,23}^{(2)}
	\end{bmatrix}
	,
	\mathbf{x}_{13,24}^{2}\!=\!\mathbf{H}_{13}^{-1}
	\begin{bmatrix}
		A_{3,24}^{(3)}\!\oplus\! D_{3,12}^{(2)}\\
		B_{1,34}^{(3)}\!\oplus\! C_{1,24}^{(3)}
	\end{bmatrix}
	\\
	&\mathbf{x}_{14,23}^{1}\!=\!\mathbf{H}_{14}^{-1}
	\begin{bmatrix}
		A_{4,23}^{(3)}\!\oplus\! B_{4,13}^{(4)}\\
		D_{1,23}^{(3)}\!\oplus\! C_{1,24}^{(4)}
	\end{bmatrix}
	,
	\mathbf{x}_{14,23}^{2}\!=\!\mathbf{H}_{14}^{-1}
	\begin{bmatrix}
		A_{4,23}^{(4)}\!\oplus\! C_{4,12}^{(4)}\\
		B_{1,34}^{(4)}\!\oplus\! D_{1,23}^{(4)}
	\end{bmatrix}
	\\
	&\mathbf{x}_{23,14}^{1}\!=\!\mathbf{H}_{23}^{-1}
	\begin{bmatrix}
		A_{3,24}^{(4)}\!\oplus\! B_{3,14}^{(3)}\\
		C_{2,14}^{(3)}\!\oplus\! D_{2,13}^{(4)}
	\end{bmatrix}
	,
	\mathbf{x}_{23,14}^{2}\!=\!\mathbf{H}_{23}^{-1}
	\begin{bmatrix}
		B_{3,14}^{(4)}\!\oplus\! D_{3,12}^{(4)}\\
		C_{2,14}^{(4)}\!\oplus\! A_{2,34}^{(4)}
	\end{bmatrix}\!\!.
\end{align*}

As we see, the delay is $\mathcal{T}_{2} = \frac{12}{24}=\frac{1}{2}$ and the DoF is $\mathcal{D}_{2} = \frac{K(1-\gamma)}{{T}_{2}} =4$. This performance is optimal.

\end{ex}


\section{Extension to the Multi-Transmitter Environment}\label{app:achiev_multiTX}
	
We now consider the multiple-transmitter case, where each of the $K_{T}$ transmitters is equipped with $L_{T}\ge 1$ antennas, and each has a cache capacity equal to a fraction $\gamma_{T}\in[ 1/K_T ,1]$ of the library, such that the transmitter-side cumulative cache size is $t_{T} =  K_{T}\gamma_T$. As we have seen, setting $L = L_{T} t_{T}$ shows how the two settings (the cache-aided MISO BC, and the corresponding multi-transmitter equivalent) share the same DoF performance $\mathcal{D}_{L}(t, C) = t+C$, irrespective of the feedback capabilities $C$.
	
The scheme for the multi-transmitter setting closely resembles the scheme in Algorithm~\ref{alg:Delivery}, with the difference being that precoding vectors $ \mathbf{h}_{\lambda\setminus\{\lambda(\ell)\}}^{\perp}$ are formed in a distributed manner. In particular, for each transmitted subfile, the $t_{T}$ transmitters who have access to that subfile must cooperate to form (each using its own $L_{T}$ antennas) a distributed precoder vector of length $L$, which possesses the attributes described in \eqref{eqPrecoderDesign}. The only modification to Algorithm~\ref{alg:Delivery} is in the precoder design (Step~$2$) where the transmission vector (cf.~\eqref{eq:transmission}) now takes the form
	\begin{align}
		\mathbf{x}_{\lambda,\pi}^{s} = \sum_{\ell=1}^{L}  \sum_{k\in \{ \lambda(\ell)\}\cup\phi_{u_{\ell}}}   \mathbf{h}_{\lambda\setminus\{\lambda(\ell)\}}^{\perp} W^{(d_{k})}_{     \{\lambda(\ell) \} \cup\pi\setminus\{k\}}.
	\end{align}
It is important to notice that for a specific $\ell\in[L]$, the respective precoder vector $\mathbf{h}^{\perp}_{\lambda\setminus\{\lambda(\ell)\}}$ is designed at the $t_T  = \frac{L}{L_{T}}$ transmitters which have stored subfile $W^{(d_{k})}_{     \{\lambda(\ell) \} \cup\pi\setminus\{k\}	}$. This further means that the precoding vectors $\mathbf{h}_{\lambda\setminus\{\lambda(\ell)\}}^{\perp}$ are subfile-dependent and thus potentially different.
	
	\paragraph*{Placement at the transmitters} To guarantee that each subfile is stored at exactly $t_T$ transmitters, we use the approach of~\cite{lampiris2018adding} which does not require an increase of the subpacketization, and which we include here for completeness. The placement algorithm starts from the first transmitter and caches the first $M_{T} = \gamma_T N$ files in their entirety, while the second transmitter caches the next set of $M_{T}$ files, and so on. Specifically, transmitter $k_{T}\in [K_{T}]$ caches
	\begin{align}\label{eqTxPlacement}
		\mathcal{Z}_{{k_{T}}}^{\text{Tx}} = \big\{ W^{(n)},\ \ n\in\{1+(k_{T}-1)M_{T}, ..., k_{T}M_{T} \}\big\},
	\end{align}
where we notice that the index of each file is calculated using the modulo operation, i.e., each file index $n\in[N]$ appearing in \eqref{eqTxPlacement} takes the form $n =  (n-1) \mod ( N) +1$.
All the other steps remain the same.

		\section{Converse}\label{secConverse}

	In this section, we prove the converse part of Theorem~\ref{theoMultiTXresult}, corresponding to the multi-transmitter environment, and, by extension, the converse part of Theorem~\ref{theoSingleTXresult}, which can be deduced by setting the problem parameters as $K_T = 1$, $L_{T}=L$, and $t_{T} = 1$.

	The bound draws partly from\cite{7857805}, mainly for the initial steps, but we introduce new ideas which allow us to capture the CSI-availability effect as well as introduce a new bounding solution for the optimization problem that directly tightens the converse. 	
	Similarly to \cite{7857805}, we are constrained to i) placement done under the assumption of uncoded prefetching, and ii) linear delivery schemes that have the one-shot property, where no data is transmitted more than once.
	
	Specifically, the steps that we implement to prove the converse part of Theorem~\ref{theoMultiTXresult} are as follows:
	\begin{enumerate}
		\item We bound the number of messages that can be simultaneously transmitted under feedback constraints.
		\item We rewrite the problem as an integer optimization problem that seeks to minimize the delivery time for a given prefetching policy and file demand vector.
		\item We obtain a novel solution of the optimization problem by leveraging the cache-size constraints and the convexity of the problem.
	\end{enumerate}
The second step, i.e., the formulation of the optimization problem, follows from\cite[Sections V.B, V.C]{7857805}, and we include it here for completeness.
On the other hand, the novelty lies on the first and the third steps, which are instrumental in obtaining the converse.

We begin by introducing some additional definitions and notation.
In the following, we consider a slightly different channel model with respect to the one described in Section~\ref{se:sys_mod} for the achievable scheme.
Let us remark that these modifications do not impact our results, and indeed they are irrelevant for the description of the achievable scheme.
On this basis, we have omitted these considerations before for the sake of clarity, and they are incorporated only in this section.

			\subsection{Preliminary definitions}\label{se:converse_0}
			We denote the superset of all the sets of caches at the users as $\zeta^{\text{Rx}}$, such that $\zeta^{\text{Rx}}\triangleq \{\Zc_1,\dotsc, \Zc_K\}$. Similarly, the superset of cached content stored at the transmitters is denoted by $\zeta^{\text{Tx}}\triangleq \{\Zc_1^{{\text{Tx}}},\dotsc,\Zc_{K_T}^{\text{Tx}}\}$.
			 We consider that every file $W\expn$ in the library $\mathcal{F}$ is divided into $F$ packets, $\{W^{(n),f}\}_{f=1}^{F}$, each of size $B$ bits.
			The caching is done at the level of packets and we do not allow breaking the packets into smaller sub-packets\footnote{
			Packets are considered to be the atomic unit of size in place of bits, such that they are big enough for the laws of Shannon to apply and the probability of decoding error to vanish as $B$ increases\cite{7857805,piovanoGDoFMISO2019TransIT}.
			{Regarding the description of the achievable scheme in Section~\ref{sec:SchemeDescription}, it can be assumed w.l.o.g. that $F$ is an integer multiple of the number of subfiles.}
			}.
			As is standard, we consider that the transmitters encode each packet $W^{(d_k),f}$ into a coded packet\footnote{Here, ``coded packet'' refers to the channel coding strategy. This notation should not be confused with coded prefetching, which is not considered in this paper.} $\tilde{W}^{(d_k)}_{s} \triangleq g(W^{(d_k)}_{s})$ of $\tilde{B}$ complex symbols using a random Gaussian coding scheme $g : \Fb_2^B \rightarrow \Cb^{\tilde{B}}$ of rate $\log P + o(\log P)$.
			We introduce in the following some definitions that are instrumental to the proof.
			
				\begin{definition}[Communication Block]\label{eq:remark_udt_block}
					A Communication Block is defined as the time required to transmit a packet --- which has size equal to the atomic unit --- to a single user, in the absence of caching and of interference. A block consists of $\frac{\tilde{B}}{\log P}$ time instants.
				\end{definition}
			Hence, for a certain demand vector $\bfd$, we consider that the transmission lasts for a set $\beta$ of communication blocks, where each block $b\in\beta$ has duration $\tilde{B}$ time slots.
			During a given communication block $b$, the transmitters send a set of packets, denoted as $\rho_{b}$, to a subset of users $\kappa_{b}\subseteq[K]$ such that every user in $\kappa_b$ desires only one packet from $\rho_b$.
			The file requested by user~$k$ is denoted as $W^{(d_k)}$, and the specific packet  of $W^{(d_k)}$ that is transmitted in this communication block~$b$ is denoted as $W^{(d_k),f_{k}}$. 
 			Note that, for sake of readability, we omit the reference to the specific communication block in which the packet is scheduled. 
			Thus, the set of transmitted packets is explicitly given by~$\rho_b = \big\{W^{(d_k),f_{k}}\big\}_{k\in \kappa_{b}}$, $f_{k}\in [F]$, $d_k\in[N]$.
					
			Furthermore, the transmitters must transmit every packet of the file $W^{(d_k)}$ that is not cached by user~$k$ throughout the $|\beta|$ communication blocks.
			The transmission will last until all the required packets are correctly received.

			{The goal of the converse is }to bound the minimum number of communication blocks required to transmit the demanded files $\{ W^{(d_1)},\dotsc, W^{(d_K)} \}$ assuming the worst-case demand vectors.
			To this end, we first consider the optimal delivery time for a given placement.
			Specifically,
				for a given prefetching policy ($\zeta^{\text{Tx}},\zeta^{\text{Rx}}$), the minimum one-shot linear delivery time achievable for the worst-case demand is defined as
					\eqm{\label{eq:dt_prefetching}
						\Tc(\zeta^{\text{Tx}},\zeta^{\text{Rx}}) \triangleq \sup_{\{d_{1}, ..., d_{K}\}} \inf_{\substack{\beta\\ \{\rho_{b}\}_{b\in \beta}}} \frac{1}{F}|\beta|.
					}						
				
			Note that the delivery time is normalized with respect to the {file}-size, such that a single unit of the delivery time corresponds to $F$ communication blocks.
			By the same token, 	we can define the worst-case optimal delivery time as follows.
			\begin{definition}[{Worst-case Delivery Time}]\label{def:delivery_time}
				{In a $K$-user fully-connected wireless network with $K_{T}$ transmitters, with $L_{T}$ antennas per transmitter, with cumulative cache size $t_{T}$ at the transmitters side and $t$ at the receivers side, and upon defining $L=L_Tt_T$, the worst-case optimal delivery time is defined as the minimum achievable one-shot linear delivery time over all caching realizations:}
		\eqm{
			\mathcal{T}^\star_{L}(t) \triangleq \inf_{\zeta^{\text{Tx}},\zeta^{\text{Rx}}} {\Tc(\zeta^{\text{Tx}},\zeta^{\text{Rx}})}.
			}			
			\end{definition}
			Further, we define the optimal (one-shot linear) DoF using the previous definition.
	
			\begin{definition}[{Optimal Degrees-of-Freedom}]\label{def:dof}
                In the cache-aided network of Definition~\ref{def:delivery_time}, the optimal one-shot linear DoF takes the form
					\eqm{
						 {\Dc}^\star_{L}(t) \triangleq \frac{K(1-\gamma)}{\mathcal{T}^\star_{L}(t)}.
					}
			\end{definition}
			
			We recall that we seek to minimize the delivery time (or, equivalently, maximize the DoF performance) under constrained feedback resources, where in each communication block, the transmitters acquire feedback only for a subset of $C$ users in total.
			Let us consider a particular communication block $b$. We denote the set of users for which there exists CSIT at communication block $b$ as $\eta_b$, $\eta_b\subseteq \kappa_b$, $|\eta_b| = C $, and its complementary set as $\eta_b^\setcomp \triangleq \kappa_{b} \backslash\eta_b$.
			Furthermore, we denote the sets of transmitters and receivers who have cached the packet $W^{(d_k),f_{k}}$ intended to receiver~$k$, as $\epsilon_{k}$ and $\delta_{k}$, respectively. 

			For some set $\alpha$, the indicator function is denoted by $\Ind_{\alpha}(k)$, to mean that $\Ind_{\alpha}(k) = 1$ if $k\in\alpha$ and $0$ otherwise.
			Accordingly, we introduce
				\eqm{\label{eq:def_iprime}
						C'_k \triangleq C + \Ind_{\eta_b^\setcomp}(k),\quad\forall k\in\kappa_{b},
				}
			such that $C'_k = C$ if $k\in\eta_b$ and $C'_k = C + 1$ if $k \notin \eta_b$.
			We will also use
				\eqm{\label{eq:def_Lk_min}
					 L_k \triangleq \min(C'_k,\, L_{T}|\epsilon_k|),
				}
			such that $L_k$ represents the minimum between the number of transmit antennas that have cached the packet intended to receiver~$k$ ($W^{(d_k),f_{k}}$) and the number of users for which there is CSIT available excluding user~$k$. In other words, $C'_k$ indicates the number of users for which the transmitters can use the CSIT so as to benefit from spatial multiplexing for packet $W^{(d_k),f_{k}}$. Further, we {introduce} the following definition.
				\begin{definition}[{Packet Order}]\label{def:order_packets}
						A packet is said to be of ``\emph{order }$(u,v)$'' if it is stored in the cache of $u$ different transmitters and $v$ different users.
				\end{definition}

			\subsection{Bounding the number of simultaneous packets}\label{se:converse_1}
			Now,
			we aim to bound the number of users that can be simultaneously served during a given communication block. 
			This bound is presented in the following lemma. 
				\begin{lemma}\label{lem:bound_packets}
					Let us consider a single communication block $b\in\beta$, where each packet of set $\rho_{b}$ is scheduled to be transmitted simultaneously to one of the users of set $\kappa_{b}$, such that $|\rho_{b}|=|\kappa_{b}|=K_{b}$.
					Assume that each transmitter has only access to the CSIT of all users of set $\eta_b\subseteq\kappa_{b}$, $|\eta_b|=C$, and that for every user~$k$, $k\in\kappa_b$, the set of users that have cached the packet intended to user~$k$ is given by~$\delta_k$.  					
					For each intended packet to be successfully decoded at the appropriate receiver, the number of simultaneously transmitted packets must satisfy
						\eqm{\label{eq:bound_packets}
								K_{b} \leq \min_{  k\in \kappa_{b}}\ \left( L_k + |\delta_k|\right).
						}
				\end{lemma}
				\begin{proof}
					The proof is relegated to Appendix~\ref{se:proof_lemma_simultaneous}.
				\end{proof}	
				\begin{corollary}\label{cor:order_packets}
					Consider some communication block~$b$.
					In order {for a packet of order $(u,v)$} to be decoded successfully, it can be transmitted simultaneously with \emph{at most} $\min(C, L_T u) + v - 1$ other packets of the same order.
				\end{corollary}			
				\begin{proof}
					The proof follows directly after substituting $|\delta_{k}|$ for $v$ and $|\epsilon_{k}|$ for $u$ in~\eqref{eq:bound_packets} of Lemma~\ref{lem:bound_packets} for every $k\in \kappa_{b}$.
					Therefore, we obtain that $K_{b} \leq \min_{k\in \kappa_{b}}\min(C'_k,\,L_{T}u) + v = \min(C,\,L_T u) + v$.
				\end{proof}	

			Next, we present the definition of \emph{feasible set of packets}, which is based on Lemma~\ref{lem:bound_packets}.
				\begin{definition}[{Feasible Sets}]\label{def:feasible_packets}
						Let a communication block $b$ be characterized by the set $\rho_{b}$ of packets to be transmitted, by the set $\kappa_{b}$ of users for whom the packets are intended, and by the set $\eta_{b}$ of users for whom there is CSIT.
						A set of packets $\rho_{b}$ selected to be transmitted at communication block $b$ is said to be \emph{feasible} if it satisfies~\eqref{eq:bound_packets} in Lemma~\ref{lem:bound_packets}, i.e., if for {every} $k\in\kappa_{b}$ it holds that
							\eqm{\label{eq:bound_packets_def}
								K_{b}\leq  L_k + |\delta_k|.
							}
				\end{definition}
		
			Consider a subset of users $\delta\subseteq[K]$ and a subset of transmitters $\epsilon\subseteq[K_T]$.
			We define
\eqm{
						\omega^{(n)}_{\epsilon,\delta} \triangleq \bigcap_{\substack{f\in[F]\\ e\in\epsilon, c\in\delta}}  \{W^{(n),f} \cap \mathcal{Z}^{\text{Tx}}_e \cap \mathcal{Z}_c\}
				}
to be the set of packets of file $W\expn$, $n\in[N]$, that are exclusively stored in the caches of the transmitters in $\epsilon$ and the users in $\delta$.
			Further, the number of packets in the set~$\omega^{(n)}_{\epsilon,\delta}$ is denoted by $a^{(n)}_{\epsilon,\delta}$.

			\subsection{Lower-bound on the number of communication blocks}\label{se:converse_2_a}
			In this section, we lower-bound the number of communication blocks that are required for a successful transmission.
			This lower bound is based on a linear program that was first stated in\cite{7857805}. The formulation of the linear program matches that of\cite{7857805}, and it is presented in Appendix~\ref{app:liner_program} for completeness.
			
			Let us consider first  a given demand vector $\bfd$	and cache-placement strategies $\zeta^{\text{Tx}}$, $\zeta^{\text{Rx}}$.
			The minimum number of communication blocks $|\beta|$ required to successfully transmit all the requested files in~$\bfd$ for the specific strategies $\zeta^{\text{Tx}}$, $\zeta^{\text{Rx}}$, is denoted as $T_\beta^\star(\zeta^{\text{Tx}},\zeta^{\text{Rx}},\bfd)$  and is rigorously defined in Appendix~\ref{app:liner_program}.

			We are interested in lower-bounding the value of $T_\beta^\star(\zeta^{\text{Tx}},\zeta^{\text{Rx}},\bfd)$ for any worst-case demand~$\bfd$.
			As shown in \cite{7857805} (see also \cite{YuMA18, piovanoGDoFMISO2019TransIT}), the solution to the optimization problem can be lower-bounded by averaging over all the possible permutations of the demand vector $\bfd$.
			Hence, for a given cache-placement strategy $\zeta^{\text{Tx}}$, $\zeta^{\text{Rx}}$
			let us define				\eqm{\label{eq:P3_init}
						\bar{T}_\beta^\star(\zeta^{\text{Tx}},\zeta^{\text{Rx}}) \triangleq \frac{1}{|\psi(N,K)|} \sum_{\bfd\in\psi({N,K})} T_\beta^\star(\zeta^{\text{Tx}},\zeta^{\text{Rx}},\bfd)
				}
			to be the average number of required communication blocks over the set of all possible worst-case demand-vectors $\bfd$. In the above, $\psi({N,K})$ denotes the set of all $K$-permutations of the library files ($N$ indices), and recall that $|\psi(N,K)| = \frac{N!}{(N-K)!}$.
								
			We focus now on lower-bounding~$\bar{T}_\beta^\star(\zeta^{\text{Tx}},\zeta^{\text{Rx}})$.
			Recalling Corollary~\ref{cor:order_packets}, a packet of order $(u,v)$ can be scheduled with at most $\min(L_T u, C) +v-1$ packets of the same order.
			Consequently, for any $\epsilon\subseteq[K_T]$,  $\delta\subseteq[K]$,  such that $|\epsilon| = u$ and $|\delta| = v$, the maximum possible multiplexing gain for any packet in set $\omega\expn_{\epsilon,\delta}$, $n\in[N]$, is $\min\big(\min(L_T u, C) +v, K\big)$.

			From this bound over the maximum multiplexing gain, we can bound the minimum number of communication blocks needed for a specific demand and cache-placement strategy. 
			Specifically, let us first note that, in order to transmit all the packets in a set $\omega^{(d_j)}_{\epsilon,\delta}$ satisfying that  $|\epsilon| = u$ and $|\delta| = v$,  we need at least ${a^{(d_j)}_{\epsilon,\delta}}/{\min(\min(C,L_T u) + v,K)}$ communication blocks. 
			Upon defining
			\eqm{\label{eq:def_guv}
				g_{u,v} \triangleq \min(C,L_T u) + v
			}
			for sake of compactness, we obtain the lower bound
				\eqm{\label{eq:lower_bound_1}
					T_\beta^\star(\zeta^{\text{Tx}},\zeta^{\text{Rx}},\bfd) \geq \sum_{v=0}^{K}\sum_{u=1}^{K_T}\sum_{j=1}^{K}\sum_{\substack{\epsilon\subseteq [K_T] \\ |\epsilon| = u}}\sum_{\substack{\delta\subseteq [K] \\ |\delta| = v \\ \delta\not\owns j}} \frac{a^{(d_j)}_{\epsilon,\delta}}{g_{u,v}}.
				}
			Incorporating~\eqref{eq:lower_bound_1} in~\eqref{eq:P3_init} yields
				\eqm{%
					 & \bar{T}_\beta^\star(\zeta^{\text{Tx}},\zeta^{\text{Rx}})  \geq  \sum_{\bfd\in\psi({N,K})} \frac{1}{|\psi(N,K)|} \nonumber \\
					 & \hspace{16ex} \times\bigg(\sum_{v=0}^{K}\sum_{u=1}^{K_T}\sum_{j=1}^{K}\sum_{\substack{\epsilon\subseteq [K_T] \\ |\epsilon| = u}}\sum_{\substack{\delta\subseteq [K] \\ |\delta| = v \\ \delta \not\owns j}} \frac{a^{(d_j)}_{\epsilon,\delta}}{g_{u,v}}\bigg) \\
						&\hspace{7ex} \geq  \sum_{v=0}^{K}\sum_{u=1}^{K_T}\sum_{j=1}^{K}\sum_{\substack{\epsilon\subseteq [K_T] \\ |\epsilon| = u}}\sum_{\substack{\delta\subseteq [K] \\ |\delta| = v \\ \delta \not\owns j}}\frac{1}{N}\sum_{n=1}^{N}\frac{a\expn_{\epsilon,\delta}}{g_{u,v}}\label{eq:proof_aver_N}\\
						&\hspace{7ex}  =  \frac{1}{N}\sum_{v=0}^{K}\sum_{u=1}^{K_T}\;\frac{1}{g_{u,v}}\sum_{j=1}^{K}\sum_{\substack{\epsilon\subseteq [K_T] \\ |\epsilon| = u}}\sum_{\substack{\delta\subseteq [K] \\ |\delta| = v \\ \delta \not\owns j}}\sum_{n=1}^{N}a\expn_{\epsilon,\delta}, \label{eq:lower_bound_5}
				}	
			where~\eqref{eq:proof_aver_N} follows since, over the set of demand-vector permutations $\psi(N,K)$, every file $W^{(n)}$ is requested by every user $j$ the same number of times.
			The last equality is obtained from a simple re-ordering of terms.
			\subsection{Tightening the lower-bound}\label{se:converse_5}
			The lower-bound in~\eqref{eq:lower_bound_5} is obtained by combining the approach in\cite{7857805} with the novel outcome of Lemma~\ref{lem:bound_packets} that accounts for the limited feedback constraint.
			Henceforth, we deviate from the approach in~\cite{7857805} so as to tighten the lower-bound.
			Let us consider the total number $a_{\epsilon,\delta}$ of packets stored exclusively at the transmitters in $\epsilon\subseteq[K_{T}]$ and the receivers in set $\delta\subseteq[K]$. This number satisfies
				\eqm{
					a_{\epsilon,\delta}\triangleq \sum_{n=1}^{N}a\expn_{\epsilon,\delta}.
				}
			Similarly, let $b_{\epsilon,v}$ denote the size of the set of packets stored exclusively by all transmitters in $\epsilon$ and at a total of $v$ receivers.
			Then,
				\eqm{\label{eq:def_b_epsilon_v}
						b_{\epsilon,v}\triangleq \sum_{\substack{\delta\subseteq [K] \\ |\delta| = v}}a_{\epsilon,\delta}.
				}		
			For a given set of transmitters $\epsilon$ and a given user-set size {$|\delta|=v$}, it follows that
				\eqm{\label{eq:equivalent_bs2}%
						\sum_{j=1}^{K}\sum_{\substack{\delta\subseteq [K] \\ |\delta| = v \\ \delta \not\owns j}}a_{\epsilon,\delta} = (K-v)b_{\epsilon,v}.
				}	
			In order to prove~\eqref{eq:equivalent_bs2}, let us consider a specific subset $\delta'\subseteq[K]$, $|\delta'| = v$.
			The number of packets cached at the transmitters of set $\epsilon$ and the users of set $\delta'$ is given by $a_{\epsilon,\delta'}$.
			For a given $j\in [K]$, the term $a_{\epsilon,\delta'}$ is included in the summation $\sum_{\delta\subseteq [K],\, |\delta| = v,\,\delta \not\owns j}a_{\epsilon,\delta}$ if and only if \emph{$j\notin\delta'$}.
			Since~\eqref{eq:equivalent_bs2} sums over all $j\in[K]$ and $|\delta'|=v$, the term $a_{\epsilon,\delta'}$ appears $K-v$ times in~\eqref{eq:equivalent_bs2}, one for each $j$ satisfying that $j\notin\delta'$.
			From the fact that this holds for any $\delta'\subseteq[K]$ with $|\delta'| = v$, and from the definition of $b_{\epsilon,v}$ in~\eqref{eq:def_b_epsilon_v}, we obtain~\eqref{eq:equivalent_bs2}.
			Applying~\eqref{eq:equivalent_bs2} into~\eqref{eq:lower_bound_5} yields
				\eqm{%
					 \bar{T}_\beta^\star(\zeta^{\text{Tx}},\zeta^{\text{Rx}})
							& \geq   \frac{1}{N}\sum_{v=0}^{K}\sum_{u=1}^{K_T}\frac{K-v}{g_{u,v}}\sum_{\substack{\epsilon\subseteq [K_T] \\ |\epsilon| = u}}b_{\epsilon,v}  \\
							& =   \frac{1}{N}\sum_{v=0}^{K}\sum_{u=1}^{K_T}\frac{K-v}{\min(C,L_{T}u) +v}b_{u,v}, \label{eq:lowerbound_t_yu}
				}				
			where in~\eqref{eq:lowerbound_t_yu} we have applied~\eqref{eq:def_guv} and defined $b_{u,v}$ to be the number of packets cached at $u$ transmitters and $v$ receivers. It is direct that
				\eqm{
						b_{u,v} \triangleq\sum_{\substack{\epsilon\subseteq [K_T] \\ |\epsilon| = u}}b_{\epsilon,v}.
				}
			Note that $\bar{T}_\beta^\star(\zeta^{\text{Tx}},\zeta^{\text{Rx}})$ represents the necessary number of communication blocks to complete the transmission.
			From the definition of delivery time in~\eqref{eq:dt_prefetching}, it follows that
				\eqm{\label{eq:change_delivery_block}
					\Tc(\zeta^{\text{Tx}},\zeta^{\text{Rx}}) =	\frac{1}{F}\bar{T}_\beta^\star(\zeta^{\text{Tx}},\zeta^{\text{Rx}}),
				}
			where~\eqref{eq:change_delivery_block} simply translates the unit of measure to consider normalization by the file size instead of the packet size.
			From~\eqref{eq:lowerbound_t_yu} and~\eqref{eq:change_delivery_block} we have that
				\eqm{%
					 \Tc(\zeta^{\text{Tx}},\zeta^{\text{Rx}})
							& \geq  \frac{1}{FN}\sum_{v=0}^{K}\sum_{u=1}^{K_T}\frac{K-v}{\min(C,L_{T}u) +v}b_{u,v}. \label{eq:lower_bound_6}							
				}		
Consequently, we have obtained a lower bound that depends only on the portion of the library that is cached at a specific number of transmitters and of receivers, irrespectively of who has stored which packet.

			For some function $c(\cdot, \cdot)$, we denote the lower convex envelope of the points \[\{\big(t_1,t_2,c(t_1,t_2)\big) | t_1,t_2 \in \{0,\,1,\,\dots,\,K\}\},\] by $\conv(c(t_1,t_2))$.
			Let us introduce the notation $c({u,v})\triangleq \frac{K-v}{\min(C,L_{T}u)+v}$.
			Since $c({u,v})$ is a decreasing sequence in $v$ and non-increasing in $u$, $\conv(c({u,v}))$ is a non-increasing and convex function~\cite{YuMA18}. 						
			Furthermore, we define the number of packets cached at $u$ transmitters (resp. $v$ receivers) as $b_{u}^{t}$ (resp. $b_{v}^{r}$), i.e.,
				\eqm{
						b_{u}^{t} 	 & \triangleq \sum_{v=0}^{K} b_{u,v}, \\
						b_{v}^{r} & \triangleq \sum_{u=1}^{K_T} b_{u,v}.
				}
			Therefore, the cache-size constraints of the considered setting can be written as
				\eqm{
						\sum_{u=1}^{K_T}\sum_{v=0}^{K} b_{u,v}
								= \sum_{u=1}^{K_T} b_{u}^{t}
								= \sum_{v=0}^{K} b_{v}^{r}
							& = NF, \label{eq:bound_max_sizeN}\\
						\sum_{v=0}^{K} v b_{v}^{r} &\leq FK \gamma N, \label{eq:bound_max_sizeMr}\\
						 \sum_{u=1}^{K_T} u b_{u}^t &\leq FK_T \gamma_T N. \label{eq:bound_max_sizeMt}
				}
  		The constraint in~\eqref{eq:bound_max_sizeN} ensures that every packet of the library is cached at some node (transmitter or receiver) in the network, while~\eqref{eq:bound_max_sizeMr} corresponds to the cache size constraint at the users, and~\eqref{eq:bound_max_sizeMt} corresponds to the cache size constraint at the transmitters.
			From the above,~\eqref{eq:lower_bound_6} can be lower-bounded as
				\eqm{
					\Tc(\zeta^{\text{Tx}},\zeta^{\text{Rx}})
							& \geq  \frac{1}{FN}\sum_{v=0}^{K}\sum_{u=1}^{K_T}\frac{K-v}{\min(C,L_{T}u) +v}b_{u,v}\label{eq:proof_end_jump} \\
							& \geq  \conv\big(c({K_T\gamma_T,\; t})\big) \label{eq:lower_bound_7}	\\
							& =  \conv\LB\frac{K(1-\gamma)}{t + \min(C, K_T\gamma_T L_{T})}\RB, 			
				}
			where~\eqref{eq:lower_bound_7} comes from exploiting the convexity of the problem and from applying Jensen's Inequality.
			The detailed proof of how to reach~\eqref{eq:lower_bound_7} from~\eqref{eq:proof_end_jump} is relegated to Appendix~\ref{se:converse_a}. 	
			Since $\Tc^\star_{L}(t,C) \triangleq \inf_{\zeta^{\text{Tx}},\zeta^{\text{Rx}}} {\Tc(\zeta^{\text{Tx}},\zeta^{\text{Rx}})}$, the converse proof of Theorem~\ref{theoMultiTXresult} is concluded. \qed

		\section{Conclusion}\label{secConclusion}							
We have characterized the optimal one-shot linear DoF of the multi-antenna cache-aided broadcast channel and its multi-transmitter equivalent, under limited feedback resources {and uncoded placement,} and we have provided a novel multi-antenna coded caching algorithm which we proved to be optimal. Our converse applies to a variety of other works, allowing the identification of their exact DoF performance.

Our results showed that achieving the maximum DoF performance may only require feedback from a limited number of users equal to the number of antennas. This further allows non-scaling feedback costs with respect to the number of users, as compared to previously known methods.

\subsection*{Various benefits of reducing feedback}
This feedback reduction has multiple beneficial effects. Firstly, reducing the feedback requirements will allow for an increase of the effective DoF, simply because a bigger fraction of the coherence period is dedicated to communicating data rather than to feedback training.
Secondly, the proposed algorithm allows for the increase of the overall number of users without a subsequent increase in feedback costs.

At the end of the day, our result makes a strong argument that caching can substantially ameliorate the well known feedback bottleneck of multi-antenna high-rate environments.

\appendices

		\section{Proof of Lemma~\ref{lem:bound_packets} }\label{se:proof_lemma_simultaneous}
		We split the proof in two disjoint cases.
		We begin with the assumption that each transmitter has access to CSIT from every user scheduled to be served in the considered communication block\footnote{While this setting is already studied in~\cite{7857805}, we will recall it here since it is a preliminary step towards the general feedback-constrained bound that we present.}.
		We will then, immediately after, introduce the CSIT constraint that says that the transmitter can only receive feedback from some $C$ users.
		
		Let us consider a single communication block.
		Without loss of generality, we assume that the $K_{b}$ served users are the first $K_{b}$ users, from $1$ to $K_{b}$, 
		and that the packets to be transmitted are $\{W^{(n),1}\}_{n=1}^{K_{b}}$.
		Under the one-shot and linear precoding assumptions, it follows that each transmitter sends a linear combination of the scheduled packets.
		In particular, the transmitted signal from a given transmitter~$j$ only carries information of the packets that it has cached, i.e., it only includes the users~$i\in[K_b]$ for which $j\in\epsilon_i$.
		We define the global beam-forming vector applied to the packet intended for user~$i$ as $\mathbf{p}_i \in \Cb^{L_{T}|\epsilon_{i}| \times 1}$, since only $|\epsilon_{i}|$ transmitters have cached ${W}^{(i),1}$.
		
			\subsection{Transmission of packets with CSIT from all scheduled users}\label{se:proof_lemma_simultaneous_case1}
			We proceed in a similar way to~\cite[Lemma 3]{7857805} {by converting} the MISO BC setting into a new MISO interference channel with $K_b$ virtual transmitters, $\{\widehat{\TX}_i\}_{i=1}^{K_b}$.
			$\widehat{\TX}_i$ has $L_{T}|\epsilon_i|$ antennas and aims to transmit  ${W}^{(i),1}$ to user~$i\in [K_b]$.
			Note that the channel of different virtual transmitters is correlated because in the real physical channel the same antenna of a certain transmitter belongs to several virtual {transmitters} \cite{7857805}.
			Let us denote the channel coefficients from virtual transmitter $\widehat{\TX}_i$ to user~$k$ by $\gv_{k,i}\in\Cb^{L_T|\epsilon_i|\times 1}$.
			Then, in an analogous way to the approach in\cite{Razaviyayn2012,7857805}, it follows that the decodability conditions that must be satisfied\footnote{Since we are restricted to linear transmission schemes, the transmission block is not successful if these conditions are not satisfied, simply because the signal-to-interference ratio would not be enough to decode the intended message (cf.~\cite{7857805,piovanoCCparallelChannelsTComm2020}).} are
				\eqm{
					\gv_{k,i}^\dagger\mathbf{p}_i &= 0			\quad \forall i,k\in[K_b] : k\notin\delta_i \label{eq:cond_zf_a}\\
					\gv_{k,k}^\dagger\mathbf{p}_k &\neq 0		\quad \forall k\in [K_b],\label{eq:cond_zf_b}
				}			
			where we recall that $\delta_i$ is the subset of users that have cached the packet intended to user~$i$.
			
			Under the assumption that the transmitters have access to the CSI of all the $K_b$ served users, we can rewrite the conditions in~\eqref{eq:cond_zf_a} as follows: 
			First, we know from~\eqref{eq:cond_zf_b} that vectors $\mathbf{p}_i$
			must have at least a non-zero coefficient, denoted by $q_i$.
			We can rotate the vector such that $\mathbf{p}_i = q_i \bP_i\begin{bsmallmatrix}1 \\ \hat{\vv}_i\end{bsmallmatrix}$,
			where $\bP$ is a permutation matrix and $\hat{\vv}_i$ has size $(L|\epsilon_i|-1 \times  1)$.
			Upon defining similarly
			$\hat{\gv}_{k,i} \triangleq \bP^{-1}_i{\gv}_{k,i} = \begin{bsmallmatrix}\hat{g}_{k,i}^{(1)} \\ \hat{\gv}_{k,i}^{(2:)}\end{bsmallmatrix}$, it follows that 
				\eqm{
					\gv_{k,i}^\dagger\mathbf{p}_i
						& = (\bP_i\hat{\gv}_{k,i})^\dagger q_i\bP_i\begin{bmatrix}1 \\[-1ex]
						\hat{\vv}_i\end{bmatrix}\\
						& = q_i\LB\hat{g}_{k,i}^{(1)\dagger} +  \hat{\gv}_{k,i}^{(2:)\dagger}\hat{\vv}_i\RB
						= 0, \label{eq:cond_zf_rotated}
				}
			where the last equality follows from~\eqref{eq:cond_zf_a}.
			
			Consider now the set $\delta_i^\setcomp\triangleq [K_b]\setminus \big( \delta_i\cup \{i\} \big)$ of served users that neither cache nor desire ${W}^{(i),1}$.
			Let $m_i\triangleq |\delta_i^\setcomp|$ and let $\delta_i^\setcomp(n)$ denote the $n$-th user of $\delta_i^\setcomp$.
			Since~\eqref{eq:cond_zf_rotated} has to hold for any $k\in\delta_i^\setcomp$, 
			we obtain the following linear system:
				\eqm{\label{eq:system_eqs_1}
					\underbrace{%
						\begin{bmatrix}
							\hat{\gv}_{\delta_i^\setcomp(1),i}^{(2:)\dagger} \\
							\hat{\gv}_{\delta_i^\setcomp(2),i}^{(2:)\dagger} \\
							\vdots \\
							\hat{\gv}_{\delta_i^\setcomp(m_i),i}^{(2:)\dagger}
						\end{bmatrix}%
					}_{%
						\bA_{\delta_i^\setcomp}
					}
					\hat{\mathbf{v}}_i =
					\underbrace{%
						\begin{bmatrix}
							\hat{g}_{\delta_i^\setcomp(1),i}^{(1)\dagger} \\
							\hat{g}_{\delta_i^\setcomp(2),i}^{(1)\dagger} \\
							\vdots \\
							\hat{g}_{\delta_i^\setcomp(m_i),i}^{(1)\dagger}
						\end{bmatrix}	
					}_{%
						\bb_{\delta_i^\setcomp}
					}								
				}
			where $\bA_{\delta_i^\setcomp}\in\Cb^{m_i\times (L_{T}|\epsilon_i |-1)}$ and $\bb_{\delta_i^\setcomp}\in\Cb^{m_i \times 1}$.
			Because~\eqref{eq:cond_zf_a} needs to be satisfied,
			{the linear system in \eqref{eq:system_eqs_1} must be solvable}
			 almost surely in order to {guarantee the successful reception} of all the messages.
			{Since} $\rank(\bA_{\delta_i^\setcomp}) = \min(m_i,\,L_{T}|\epsilon_i|-1)$, {it follows} that $m_i\leq L_{T}|\epsilon_i|-1$, {implying} that
				\eqm{\label{eq:ineq_num_users_perf}
						K_{b}- |\delta_i| - 1 \leq L_{T}|\epsilon_i |-1 \  \Rightarrow \  K_{b} \leq L_{T}|\epsilon_i| + |\delta_i|.
				}
			{Let us consider now an arbitrary communication block~$b$ during which a set of users $\kappa_b$, $|\kappa_b|=K_b$, is served.} Given that~\eqref{eq:ineq_num_users_perf} must hold for any $i\in\kappa_{b}$, we obtain that
				\eqm{
						K_{b} \leq \min_{i\in\kappa_{b} } L_{T}|\epsilon_i | + |\delta_i|,
				}
			from which we obtain Lemma~\ref{lem:bound_packets} for the case without feedback constraints.
			
			\subsection{Transmission of packets with CSIT from only~$C$ users}\label{se:proof_lemma_simultaneous_case2}
					
			Let us now consider the case in which the transmitters have CSIT only from a subset $\eta$ of $|\eta|=C$ users, and recall that  $\eta^\setcomp = \kappa_{b}\backslash\eta$.
			
			Since the transmitters do not know the channel towards the users belonging to the set $\eta^\setcomp$, the condition in~\eqref{eq:cond_zf_a} can not be satisfied with a high probability.
			In consequence, the transmitters can only use the CSIT from the users belonging to $\eta$, such that the solvable linear system becomes
				\eqm{\label{eq:system_eqs_2}
					\bA_{\delta_i^\setcomp\cap\eta}
					\hat{\vv}_i =
					\bb_{\delta_i^\setcomp\cap\eta},
				}		
			where $\bA_{\delta_i^\setcomp\cap\eta}$ and $\bb_{\delta_i^\setcomp\cap\eta}$ are defined just like $\bA_{\delta_i^\setcomp}$ and $\bb_{\delta_i^\setcomp}$ in~\eqref{eq:system_eqs_1} except that now we only consider the users in~$\{\delta_i^\setcomp\cap\eta\}$ rather than in $\delta_i^\setcomp$.
			Note that the set $\delta_i^\setcomp\cap\eta$ is comprised of the users that have not cached the message for user~$i$ \emph{and} for whom the transmitter has acquired CSIT.
			
			We focus now on the {required conditions that allow the successful reception of packets by} each user in $\kappa_{b}$.
			From~\eqref{eq:system_eqs_2}, it follows that the set $\delta_i^\setcomp$ must be a subset of the users for which there is CSIT available ($\delta_i^\setcomp\subseteq\eta$) for any user $i\in\kappa_{b}$.
			This is due to the fact that, for any user~$i$ such that $i\notin\eta$, the lack of CSIT implies the impossibility of satisfying~\eqref{eq:cond_zf_a} and thus the impossibility of correctly decoding at user~$i$\cite{piovanoCCparallelChannelsTComm2020}.
			Following the same reasoning as in~\eqref{eq:ineq_num_users_perf}, having $\delta_i^\setcomp\subseteq\eta$ implies that $m_i \leq C$.
			However, note that it holds that $i\notin \delta_i^\setcomp$ for any user~$i$. 
			Thus, $|\delta_i^\setcomp\cap\eta|\leq C-1$ \,$\forall i \in\eta$.
			Let $m^{\eta}_i\triangleq |\delta_i^\setcomp\cap\eta|$ be the size of the intersection between the set of users not caching the packet intended for user~$i$ and the set of users for whom there is CSIT available.
			We then have
				\eqm{\label{eq:ineq_num_users_notperf2}
					&&m^{\eta}_i &\leq C - 1
						& \Rightarrow \quad	K_{b}  &\leq C + |\delta_i|
						&& \text{ if $i \in\eta$,}
						\\
					&&m^{\eta}_i &\leq C
						& \Rightarrow \quad	K_{b}  & \leq C + |\delta_i|+ 1
						&& \text{ if $i \notin\eta$.}
				}	
			We recall the notation introduced in~\eqref{eq:def_iprime}, where for any $i\in\kappa_{b}$ we define $C'_i \triangleq C + \Ind_{\eta^\setcomp}(i)$. 
			Furthermore, the bound in~\eqref{eq:ineq_num_users_perf} also holds, as it suffices to consider a genie that provides the CSIT of every user to the transmitters.
				Since~\eqref{eq:ineq_num_users_perf} must hold for any $i\in\kappa_{b}$, we obtain that
				\eqm{
						K_{b} \leq \min_{i\in\kappa_{b}} \Big(\min(C'_i,\, L_{T}|\epsilon_i|) + |\delta_i|\Big),
				}
			which concludes the proof of Lemma~\ref{lem:bound_packets}. \qed

		\section{Additional Proofs and Material}	
		\subsection{Integer Program Formulation}\label{app:liner_program}	
													
			Let us consider a given demand vector $\bfd$
			and cache-placement strategies $\zeta^{\text{Tx}}$, $\zeta^{\text{Rx}}$ at the transmitters and the receivers, respectively.
			Equipped with the definition of a feasible set of packets (cf. Definition~\ref{def:feasible_packets}), we write an integer program that seeks to minimize the number of required communication blocks for a specific $\zeta^{\text{Tx}}$, $\zeta^{\text{Rx}}$, $\bfd$, as follows.
				\eqmv{P1}{%
					\min_{\substack{\rho_b\\ b\in\beta}}\quad & |\beta| \tag{P1-a}\\
								& \text{s.t. }	\bigcup_{b\in\beta} \rho_b = \bigcup_{k\in[K]} \LB W^{(d_{k})}\backslash \mathcal{Z}_{k} \RB \tag{P1-b} \label{eq:p1a}\\
								& \qquad	 \rho_{b} \text{ is feasible }	\forall b\in \beta \label{eq:p1c}\tag{P1-c} ,
				}
			where~\eqref{eq:p1a} imposes the necessary condition that all the demanded packets that are not in the cache of the intended user must be transmitted.
			The solution to~\eqref{eq:P1} is denoted by~$T_\beta^\star(\zeta^{\text{Tx}},\zeta^{\text{Rx}},\bfd)$.

		\subsection{{Transition from~\eqref{eq:proof_end_jump} to \eqref{eq:lower_bound_7} }}\label{se:converse_a}	
		We {have }that
			\eqm{
				&\frac{1}{FN}\sum_{v=0}^{K}\sum_{u=1}^{K_T}b_{u,v} \conv\LB\frac{K-v}{\min(C,L_{T}u) +v}\RB\nonumber\\
						&\qquad\qquad =	\frac{1}{NF}\sum_{v=0}^{K}\sum_{u=1}^{K_T} b_{u,v}  \conv\LB c({u,v})\RB \\
						&\qquad\qquad \myoverset{(a)}{=} \frac{\sum_{v=0}^{K}\sum_{u=1}^{K_T} b_{u,v} \conv\LB c({u,v})\RB}{\sum_{v=0}^{K}\sum_{u=1}^{K_T} b_{u,v}} \\
						&\qquad\qquad  \myoverset{(b)}{\geq} \conv \LB c{\bigg(
									\frac{\sum_{u=1}^{K_T} u b^t_{u}}{\sum_{u=1}^{K_T} b^t_u},\
									\frac{\sum_{v=0}^{K} v b^r_{v}}{\sum_{v=0}^{K} b^r_v}
							\bigg)}\RB \\
						&\qquad\qquad \myoverset{(c)}{=}  \conv \LB c{\bigg(
								\frac{\sum_{u=1}^{K_T} u b^t_u}{NF}, \
								\frac{\sum_{v=0}^{K} v b^r_v}{NF}
							\bigg)}\RB,
			}
		where $(a)$ comes from~\eqref{eq:bound_max_sizeN}, $(b)$ from Jensen's Inequality,
		and $(c)$ from~\eqref{eq:bound_max_sizeN} again.
		The monotonically decreasing nature of $c(u,v)$, combined with~\eqref{eq:bound_max_sizeMr}-\eqref{eq:bound_max_sizeMt}, yield
			\eqm{
					&\conv \LB c{\bigg(\frac{\sum_{u=1}^{K_T} u b^t_u}{NF}, \
								\frac{\sum_{v=0}^{K} v b^r_v}{NF}							
							\bigg)}\RB \nonumber\\
						&\hspace{12ex} \geq  \conv \LB c{\bigg(
								\frac{F K_T \gamma_T N}{NF}, \
								\frac{F K \gamma N}{NF}
							\bigg)}\RB \\
							&\hspace{12ex} =  \conv \big( c(t_T , t)\big). 
			}
		By recovering~\eqref{eq:proof_end_jump}, we can write that
				\eqm{%
					 \Tc(\zeta^{\text{Tx}},\zeta^{\text{Rx}})
							& \geq  \frac{1}{FN}\sum_{v=0}^{K}\sum_{u=1}^{K_T}b_{u,v}\conv\LB\frac{K-v}{\min(C,L_{T}u) +v}\RB\nonumber \\
							& \geq \conv\LB\frac{K(1-\gamma)}{\min(C,L_{T} t_T) + t}\RB, \label{eq:lower_bound_concluding}							
				}			
		which concludes the proof. \qed

		\subsection{Discussion on the CSI acquisition}\label{sec:csi}

The CSIT acquisition phase can be done in a standard way {such that the $L$ users} (set $\lambda$) communicate pilots, which will allow the transmitter to estimate the channels of these users. As such, here we focus on the process of CSIR acquisition where the goal is to communicate at each user of set $\pi\cup \lambda$ the channel-precoder products. The process requires $1$ training slot for each precoder, which amounts to $L$ training slots per transmission.

Communication of precoder $\mathbf{h}^{\perp}_{\mu}, \mu\subset[L], |\mu| = L-1$, takes the form
\begin{equation}
	\mathbf{x}_{\mu} =
	\begin{bmatrix}
		 \mathbf{h}_{\mu}^{\perp}(1) \mathbf{s}(1)\\
		 \vdots\\
		  \mathbf{h}_{\mu}^{\perp}(L)  \mathbf{s}(L)
	\end{bmatrix}
\end{equation}
where $\mathbf{s}$ denotes a single training vector.

The received message, ignoring the noise for simplicity, at some user $k\in[K]$, takes the form
\begin{equation}
	y_{k} = \mathbf{h}^{\dagger}_{k} \mathbf{x}_{\mu}  = \sum_{\ell = 1}^{L}  \mathbf{h}^{\dagger}_{k}(\ell)\mathbf{h}_{\mu}^{\perp}(\ell)  \mathbf{s}(\ell)
\end{equation}
from which the composite channel-precoder product $ \mathbf{h}^{\dagger}_{k} \mathbf{h}_{\mu}^{\perp}$ can be calculated.

To summarize, CSIT requires $L$ slots because only the $L$ users need to transmit their channel state, and global CSIR requires $L$ slots because, for each fixed precoder, one training symbol suffices to communicate the composite channel-precoder product to any number of users.

\subsection{Extensive example for the scheme}

We conclude with a more involved example that aims to help the reader gain a deeper understanding of the mechanics of our algorithm.

\begin{ex}
Let us consider the $L=2$ MISO BC with $K=6$ users and cumulative cache size $t= 4$. Then, the required $30$ transmissions to satisfy all users' demands are
\begin{align*}
	&\mathbf{x}_{12,3456}^{1}=\mathbf{H}^{-1}_{12}
	\begin{bmatrix}
		A_{2,3456}^{(1)}\oplus C_{2,1456}^{(1)}\oplus D_{2,1356}^{(1)}\\
		B_{1,3456}^{(1)}\oplus E_{1,2346}^{(1)}\oplus F_{1,2345}^{(1)}
	\end{bmatrix}
 , \\ 	
 &\mathbf{x}_{12,3456}^{2}=\mathbf{H}^{-1}_{12}
	\begin{bmatrix}
		A_{2,3456}^{(2)}\oplus E_{2,1346}^{(1)}\oplus F_{2,1345}^{(1)}\\
		B_{1,3456}^{(2)}\oplus C_{1,2456}^{(1)}\oplus D_{1,2356}^{(1)}
	\end{bmatrix}
 , \\ 	
	&\mathbf{x}_{13,2456}^{1}=\mathbf{H}^{-1}_{13}
	\begin{bmatrix}
		A_{3,2456}^{(1)} \oplus B_{3,1456}^{(1)} \oplus D_{3,1256}^{(1)}\\
		C_{1,2456}^{(2)}\oplus E_{1,2346}^{(2)}\oplus F_{1,2345}^{(2)}
	\end{bmatrix}
 , \\
 	&\mathbf{x}_{13,2456}^{2}=\mathbf{H}^{-1}_{13}
	\begin{bmatrix}
		A_{3,2456}^{(2)}\oplus E_{3,1246}^{(1)}\oplus F_{3,1245}^{(1)}\\
		C_{1,2456}^{(3)}\oplus B_{1,3456}^{(3)}\oplus D_{1,2356}^{(2)}
	\end{bmatrix}
 , \\ 	
	&\mathbf{x}_{14, 2356}^{1}=\mathbf{H}^{-1}_{14}
	\begin{bmatrix}
		A_{4,2356}^{(1)}\oplus B_{4,1356}^{(1)}\oplus C_{4,1256}^{(1)}\\
		D_{1,2356}^{(3)}\oplus E_{1,2346}^{(3)}\oplus F_{1,2345}^{(3)}
	\end{bmatrix}
, \\
	&\mathbf{x}_{14,2356}^{2}=\mathbf{H}^{-1}_{14}
	\begin{bmatrix}
		A_{4,2356}^{(2)}\oplus E_{4,1236}^{(1)}\oplus F_{4,1235}^{(1)}\\
		D_{1,2356}^{(4)}\oplus B_{1,3456}^{(4)}\oplus C_{1,2456}^{(4)}
	\end{bmatrix}
 , \\ 	
	&\mathbf{x}_{15, 2346}^{1}=\mathbf{H}^{-1}_{15}
	\begin{bmatrix}
		A_{5,2346}^{(1)}\oplus B_{5,1346}^{(1)}\oplus C_{5,1246}^{(1)}\\
		E_{1,2346}^{(4)}\oplus D_{1,2356}^{(5)}\oplus F_{1,2345}^{(4)}
	\end{bmatrix}
, \\
	&\mathbf{x}_{15,2346}^{2}=\mathbf{H}^{-1}_{15}
	\begin{bmatrix}
		A_{5,2346}^{(2)}\oplus D_{5,1236}^{(1)}\oplus F_{5,1234}^{(1)}\\
		E_{1,2346}^{(5)}\oplus B_{1,3456}^{(5)}\oplus C_{1,2456}^{(5)}
	\end{bmatrix}
 , \\ 	
	&\mathbf{x}_{16,2345}^{1}=\mathbf{H}^{-1}_{16}
	\begin{bmatrix}
		A_{6,2345}^{(1)}\oplus B_{6,1345}^{(1)}\oplus C_{6,1245}^{(1)}\\
		F_{1,2345}^{(5)}\oplus D_{1,2356}^{(6)}\oplus E_{1,2346}^{(6)}
	\end{bmatrix}
, \\
	&\mathbf{x}_{16, 2345}^{2}=\mathbf{H}^{-1}_{16}
	\begin{bmatrix}
		A_{6,2345}^{(2)}\oplus D_{6,1235}^{(1)}\oplus E_{6,1234}^{(1)}\\
		F_{1,2345}^{(6)}\oplus B_{1,3456}^{(6)}\oplus C_{1,2456}^{(6)}
	\end{bmatrix}
 , \\ 	
	&\mathbf{x}_{23, 1456}^{1}=\mathbf{H}^{-1}_{23}
	\begin{bmatrix}
		B_{3,1456}^{(2)}\oplus A_{3,2456}^{(3)}\oplus D_{3,1256}^{(2)}\\
		C_{2,1456}^{(2)}\oplus E_{2,1346}^{(2)}\oplus F_{2,1345}^{(2)}
	\end{bmatrix}
, \\
	&\mathbf{x}_{23,1456}^{2}=\mathbf{H}^{-1}_{23}
	\begin{bmatrix}
		B_{3,1456}^{(3)}\oplus E_{3,1246}^{(2)}\oplus F_{3,1245}^{(2)}\\
		C_{2,1456}^{(3)}\oplus A_{2,3456}^{(3)}\oplus D_{2,1356}^{(2)}
	\end{bmatrix}
 , \\ 	
	&\mathbf{x}_{24, 1356}^{1}=\mathbf{H}^{-1}_{24}
	\begin{bmatrix}
		B_{4,1356}^{(2)}\oplus A_{4,2356}^{(3)} \oplus C_{4,1256}^{(2)}\\
		D_{2,1356}^{(3)}\oplus E_{2,1346}^{(3)}\oplus F_{2,1345}^{(3)}
	\end{bmatrix}
, \\
	&\mathbf{x}_{24, 1356}^{2}=\mathbf{H}^{-1}_{24}
	\begin{bmatrix}
		B_{4,1356}^{(3)}\oplus E_{4,1236}^{(2)}\oplus F_{4,1235}^{(2)}\\
		D_{2,1356}^{(4)}\oplus A_{2,3456}^{(4)}\oplus C_{2,1456}^{(4)}
	\end{bmatrix}
 , \\ 	
	&\mathbf{x}_{25, 1346}^{1}=\mathbf{H}^{-1}_{25}
	\begin{bmatrix}
		B_{5,1346}^{(2)}\oplus A_{5,2346}^{(3)}\oplus C_{5,1246}^{(2)}\\
		E_{2,1346}^{(4)}\oplus D_{2,1356}^{(5)}\oplus F_{2,1345}^{(4)}
	\end{bmatrix}
, \\
	&\mathbf{x}_{25, 1346}^{2}=\mathbf{H}^{-1}_{25}
	\begin{bmatrix}
		B_{5,1346}^{(3)}\oplus D_{5,1236}^{(2)}\oplus F_{5,1234}^{(2)}\\
		E_{2,1346}^{(5)}\oplus A_{2,3456}^{(5)}\oplus C_{2,1456}^{(5)}
	\end{bmatrix}
 , \\ 	
	&\mathbf{x}_{26, 1345}^{1}=\mathbf{H}^{-1}_{26}
	\begin{bmatrix}
		B_{6,1345}^{(2)}\oplus A_{6,2345}^{(3)}\oplus C_{6,1245}^{(2)}\\
		F_{2,1345}^{(5)}\oplus D_{2,1356}^{(6)}\oplus E_{2,1346}^{(6)}
	\end{bmatrix}
, \\
	&\mathbf{x}_{26, 1345}^{2}=\mathbf{H}^{-1}_{26}
	\begin{bmatrix}
		B_{6,1345}^{(3)}\oplus D_{6,1235}^{(2)}\oplus E_{6,1234}^{(2)}\\
		F_{2,1345}^{(6)}\oplus A_{2,3456}^{(6)}\oplus C_{2, 1456}^{(6)}
	\end{bmatrix}
 , \\ 	
	&\mathbf{x}_{34, 1256}^{1}=\mathbf{H}^{-1}_{34}
	\begin{bmatrix}
		C_{4,1256}^{(3)}\oplus A_{4,2356}^{(4)}\oplus B_{4,1356}^{(4)}\\
		D_{3,1256}^{(3)}\oplus E_{3,1246}^{(3)}\oplus F_{3,1245}^{(3)}
	\end{bmatrix}
, \\
	&\mathbf{x}_{34, 1256}^{2}=\mathbf{H}^{-1}_{34}
	\begin{bmatrix}
		C_{4,1256}^{(4)}\oplus E_{4,1236}^{(3)}\oplus F_{4,1235}^{(3)}\\
		D_{3,1256}^{(4)}\oplus A_{3,2456}^{(4)}\oplus B_{3,1456}^{(4)}
	\end{bmatrix}
 , \\ 	
	&\mathbf{x}_{35, 1246}^{1}=\mathbf{H}^{-1}_{35}
	\begin{bmatrix}
		C_{5,1246}^{(3)}\oplus A_{5,2346}^{(4)}\oplus B_{5,1346}^{(4)}\\
		E_{3,1246}^{(4)}\oplus D_{3,1256}^{(5)}\oplus F_{3,1245}^{(4)}
	\end{bmatrix}
, \\
	&\mathbf{x}_{35, 1246}^{2}=\mathbf{H}^{-1}_{35}
	\begin{bmatrix}
		C_{5,1246}^{(4)}\oplus D_{5,1236}^{(3)}\oplus F_{5,1234}^{(3)}\\
		E_{3,1246}^{(5)}\oplus A_{3,2456}^{(5)}\oplus B_{3,1456}^{(5)}
	\end{bmatrix}
 , \\ 	
	&\mathbf{x}_{36, 1245}^{1}=\mathbf{H}^{-1}_{36}
	\begin{bmatrix}
		C_{6,1245}^{(3)}\oplus A_{6,2345}^{(4)}\oplus B_{6,1345}^{(4)}\\
		F_{3,1245}^{(5)}\oplus D_{3,1256}^{(6)}\oplus E_{3,1246}^{(6)}
	\end{bmatrix}
, \\
	&\mathbf{x}_{36, 1245}^{2}=\mathbf{H}^{-1}_{36}
	\begin{bmatrix}
		C_{6,1245}^{(4)}\oplus D_{6,1235}^{(3)}\oplus E_{6,1234}^{(3)}\\
		F_{3,1245}^{(6)}\oplus A_{3,2456}^{(6)}\oplus B_{3,1456}^{(6)}
	\end{bmatrix}
 , \\ 	
	&\mathbf{x}_{45, 1236}^{1}=\mathbf{H}^{-1}_{45}
	\begin{bmatrix}
		D_{5,1236}^{(4)}\oplus A_{5,2346}^{(5)}\oplus B_{5,1346}^{(5)}\\
		E_{4,1236}^{(4)}\oplus C_{4,1256}^{(5)}\oplus F_{4, 1235}^{(4)}
	\end{bmatrix}
, \\
	&\mathbf{x}_{45,1236}^{2}=\mathbf{H}^{-1}_{45}
	\begin{bmatrix}
		D_{5,1236}^{(5)}\oplus C_{5,1246}^{(5)}\oplus F_{5,1234}^{(4)}\\
		E_{4,1236}^{(5)}\oplus A_{4,2356}^{(5)}\oplus B_{4,1356}^{(5)}
	\end{bmatrix}
 , \\ 	
	&\mathbf{x}_{46,1235}^{1}=\mathbf{H}^{-1}_{46}
	\begin{bmatrix}
		D_{6,1235}^{(4)}\oplus A_{6,2345}^{(5)}\oplus B_{6,1345}^{(5)}\\
		F_{4,1235}^{(5)}\oplus C_{4,1256}^{(6)}\oplus E_{4,1236}^{(6)}
	\end{bmatrix}
 , \\
	&\mathbf{x}_{46, 1235}^{2}=\mathbf{H}^{-1}_{46}
	\begin{bmatrix}
		D_{6,1235}^{(5)}\oplus C_{6,1245}^{(5)}\oplus E_{6,1234}^{(4)}\\
		F_{4,1235}^{(6)}\oplus A_{4,2356}^{(6)}\oplus B_{4,1356}^{(6)}
	\end{bmatrix}
 , \\ 	
	&\mathbf{x}_{56,1234}^{1}=\mathbf{H}^{-1}_{56}
	\begin{bmatrix}
		E_{6,1234}^{(5)}\oplus A_{6,2345}^{(6)}\oplus B_{6,1345}^{(6)}\\
		F_{5,1234}^{(5)}\oplus C_{5,1246}^{(6)}\oplus D_{5,1236}^{(6)}
	\end{bmatrix}
, \\
	&\mathbf{x}_{56,1234}^{2}=\mathbf{H}^{-1}_{56}
	\begin{bmatrix}
		E_{6,1234}^{(6)}\oplus C_{6,1245}^{(6)}\oplus D_{6,1235}^{(6)}\\
		F_{5,1234}^{(6)}\oplus A_{5,2346}^{(6)}\oplus B_{5,1346}^{(6)}
	\end{bmatrix}.
\end{align*}

By examining any of the above transmitted vectors, we can deduce that each transmission serves a total of $6$ users, with a feedback cost of $C=2$.

\end{ex}
			

\enlargethispage{-1.2cm}

\IEEEtriggeratref{3}

\end{document}

%% file: Definitions.tex
	\newcommand{\eqmv}[2]{\begin{varsubequations}{#1}\label{eq:#1}\begin{align}#2\end{align}\end{varsubequations}}

	\definecolor{mycolor1}{rgb}{0.00000,0.44700,0.74100}%
	\definecolor{mycolor2}{rgb}{0.85000,0.32500,0.09800}%
	\definecolor{mycolor3}{rgb}{0.92900,0.69400,0.12500}%
	\definecolor{mycolor4}{rgb}{0.49400,0.18400,0.55600}%
	\definecolor{mycolor5}{rgb}{0.300000,0.64314,0.00000}%
	\definecolor{gray}{rgb}{0.6,0.6,0.6}%
	\colorlet{mygreen}{green!60!black}%
	\colorlet{myblue}{mycolor1!80!black}%

	\newcommand\itemEq[1][]{%
		\ifx\relax#1\relax  \item \else \item[#1] \fi
		\abovedisplayskip=0pt\abovedisplayshortskip=0pt~\vspace*{-\baselineskip}}
	
	\newcommand{\myoverset}[2]{\ensuremath{\overset{\mathclap{#1}}{#2}}}

\newcommand{\itb}{\begin{itemize}}
\newcommand{\ite}{\end{itemize}}
\newcommand{\enb}{\begin{enumerate}}
\newcommand{\ene}{\end{enumerate}}
\newcommand{\eqm}[1]{\begin{align}#1\end{align}}

	\newcommand{\setcomp}{{\mathsf{c}}}

	\newcommand{\LB}{\left(}
	\newcommand{\RB}{\right)}

	\newcommand{\expn}{^{(n)}}

\DeclareMathAlphabet{\mathbit}{OML}{cmr}{bx}{it}
\DeclareMathAlphabet{\mathsf}{OT1}{cmss}{m}{n}
\DeclareMathAlphabet{\mathTXf}{OT1}{cmss}{bx}{it}

\DeclareMathOperator{\rank}{rank}

	\DeclareMathOperator{\TX}{TX}

	\newcommand{\Cb}{{{\mathbb{C}}}}
	\newcommand{\Fb}{{{\mathbb{F}}}}

	\newcommand{\Ind}{{\mathds{1}}} 

\newcommand{\gv}{\mathbf{g}}

\newcommand{\vv}{\mathbf{v}}

\newcommand{\Dc}{{{\mathcal{D}}}}

\newcommand{\Tc}{{{\mathcal{T}}}}

\newcommand{\Zc}{{{\mathcal{Z}}}}

\newcommand{\bA}{\mathbf{A}}

\newcommand{\bP}{\mathbf{P}}

\newcommand{\bb}{\bm{b}}

	\newcommand{\conv}{{\text{conv}}}